\documentclass[12pt, draftclsnofoot, onecolumn]{IEEEtran}
\usepackage{amsfonts}
\usepackage{dsfont}
\usepackage{floatrow} 
\usepackage{setspace}
\usepackage{color}
\usepackage[table]{xcolor}
\usepackage{amssymb}
\usepackage{cite}
\usepackage[cmex10]{amsmath}
\usepackage{algorithm}
\usepackage{algorithmic} 
\usepackage{array}
\usepackage{mathrsfs}
\usepackage{graphicx}
\usepackage{latexsym}
\usepackage{amscd}
\usepackage{amsfonts}
\usepackage[font=small]{caption,subcaption}
\usepackage{amsmath,amscd,amssymb,verbatim}
\usepackage{graphics}
\usepackage{amsthm}
\usepackage[T1]{fontenc}
\usepackage[utf8]{inputenc}
\usepackage{authblk}
\usepackage{optidef}
\usepackage{mathtools}
\usepackage{soul}
\usepackage{anyfontsize}
\usepackage{textcomp}
\usepackage{xcolor}
\usepackage{datetime}
\usepackage{microtype}

\usepackage{blindtext}
\IEEEoverridecommandlockouts

\newtheorem{theorem}{Lemma}
\newtheorem{theorem1}{Corollary}

\begin{document}
%
\renewcommand{\thefootnote}{\fnsymbol{footnote}}
\title{Joint Convexity of  Error Probability in Blocklength and  Transmit Power in the Finite Blocklength Regime 
\thanks{
This work was supported in part by  the China National Key Research and Development Program under Grant 2021YFB2900301, in part by BMBF Germany in the program of ``Souverän. Digital. Vernetzt.'' Joint Project 6G-RIC with project identification number 16KISK028, and in part by the U.S. National Science Foundation under Grant CNS-2128448.  

Y. Zhu, Y. Hu, X. Yuan and A. Schmeink are with the Chair of Information Theory and Data Analytics, RWTH Aachen University, 52074 Aachen, Germany (e-mail: \{zhu,hu,yuan,schmeink\}@isek.rwth-aachen.de). M. C. Gursoy is with the Department of Electrical Engineering and
Computer Science, Syracuse University, NY 13210, USA (e-mail: mcgursoy@syr.edu). 
H. V. Poor is     with the Department of Electrical Engineering, Princeton University, Princeton,  NJ 08544 USA
(e-mail:poor@princeton.edu). 
*Y. Hu is the corresponding  author. } }
\author{
Yao Zhu, 
Yulin Hu$^*$, 
Xiaopeng Yuan, 
M. Cenk Gursoy, 
H. Vincent Poor,
Anke~Schmeink 
\vspace*{-1.5cm} 
}
\maketitle
    
\vspace{-5pt}
\begin{abstract}
To support ultra-reliable and low-latency services for  mission-critical applications,
transmissions are usually carried via short blocklength codes, i.e., in the so-called 
finite blocklength (FBL) regime. 
Different from the infinite blocklength regime where transmissions are assumed to be arbitrarily reliable at the Shannon's capacity, the reliability and capacity performances  of an FBL transmission are impacted by the coding blocklength.  
The relationship among reliability, coding rate, blocklength and channel quality has recently been characterized in the literature, considering the 
FBL performance model. 
In this paper, we follow this model,  and prove the joint convexity of the FBL error probability with respect to blocklength and transmit power within a region of interest, as a key enabler for designing systems to achieve globally optimal performance levels. Moreover, we apply the joint convexity to general use cases and efficiently solve the joint optimization problem in the setting with multiple users. We also extend the applicability of the proposed approach by proving that the joint convexity still holds in fading channels, as well as in relaying networks. Via simulations, we validate our analytical results and demonstrate the advantage of leveraging the joint convexity compared to other commonly-applied approaches.
\end{abstract}
\vspace*{-0.12cm}
\vspace{-10pt}

\section{introduction}

{Supporting   low-latency and ultra-reliability transmissions  is crucial for wireless networks in order to enable novel mission-critical applications in real-time, such as vehicle-to-vehicle/infrastructure communications, augmented/virtual reality and factory automation~\cite{Bennis_tail,Meng2022Sampling}. 
One of the key approaches to fulfill the stringent latency requirements of those applications is to employ short-packet communications carried out via finite blocklength (FBL) codes~\cite{URLLC_intro, Chen_MA}.} 
In such FBL scenarios, the transmission error   is no longer negligible, as a key departure from the so-called infinite blocklength (IBL) regime in which the transmissions are  arbitrarily reliable at the rate of Shannon's capacity. 
In 2010, Polyanskiy~\emph{et al.},  in the landmark work~\cite{Polyanskiy_2010},  derived a closed-form expression of the decoding error probability in the FBL regime, when transmissions occur over a channel with additive white Gaussian noise (AWGN). It has been shown that the decoding error probability of FBL codes is positive even if the coding rate is less than the Shannon capacity, and this probability is higher when the blocklength gets shorter. 
Subsequently, the FBL performance characterization in~\cite{Polyanskiy_2010} has been extended to Gilbert-Elliott 
channels~\cite{Polyanskiy_2011}, and quasi-static flat-fading 
channels~\cite{Yang_2014,Yang_ISIT2014}. 
{More recently, the FBL performance
has been analytically investigated considering cooperative relay networks~\cite{Hu_2015,Hu_letter_2015,Hu_TWC_2015,Li_2016}, non-orthogonal multiple access (NOMA) schemes~\cite{Ding_NOMA,Lai_NOMA_analysis,Ghanami_NOMA_analysis}, wireless power transfer~\cite{Alcaraz_WP_linear_approx_app2,Makki_2016_WP_linear_app}, system security~\cite{Chao_security} and green communications~\cite{Singh_2021_EE_V_1, Zhu_2022_EE}.}

{Additionally, the design of future Internet-of-Things (IoT) networks calls for the support of massive access, also known as massive Machine Type Communications (mMTC), while guaranteeing the stringent latency and reliability requirements of each device~\cite{Chen_MA}. Therefore, the short packet communications are likely to be carried out with FBL codes. Following the FBL model in~\cite{Polyanskiy_2010} to characterize the accurate error probability performance,} 
{ a set of optimal system designs has been provided for such low-latency multi-device   networks 
via  blocklength allocation~\cite{Abdelsadek_blocklength_single, Sun_joint_optimization_alt}, power control~\cite{Hu_power_control,He_2021_power_SCA, Yuan_2022_power_approx} and transmission rate selection~\cite{Zheng_transmission_rate}. All the above designs address optimality with respect 
to a single factor/parameter, e.g., the  blocklength, transmit power or the coding rate.} 
However, it is practically more important to investigate multi-factor joint designs, in which, for instance, joint resource allocation  with respect to both transmit power and blocklength is considered. A joint design via the optimization of multiple parameters 
leads to improved performance compared to single-factor design, and helps the network meet the critical requirements of different devices.  
Moreover, from a theoretical point of view, in comparison  to a single factor optimization, a joint design  provides greater insights into the fundamental tradeoff among the considered factors and parameters especially in a multi-user scenario. Understanding such tradeoffs is particularly important  in designing resource-limited networks (e.g., in terms of coding blocklength and transmit power) supporting 
multi-device low latency  and ultra reliable  transmissions.

A key challenge in obtaining the optimal solution of such a joint design  is 
to characterize the principal properties (e.g., the joint convexity) 
of the design objective which 
involves complicated  normal approximations~\cite{Polyanskiy_2010}. 
In fact, the widely applied  FBL performance metrics introduced in~\cite{Polyanskiy_2010} have two  characteristics that complicate the design and analysis. First, the  expressions   for the    coding rate and     error probability   
may not apply to extreme scenarios, e.g., with extremely  low SNR or short blocklength. 
Secondly, the expressions for the FBL coding rate and error probability are  significantly more complicated than the corresponding expression for the Shannon capacity, making model-based analytical derivations challenging.   
To address the first issue, existing studies usually 
assume  that the values of the SNR and blocklength lie within certain regions of practical interest  where the expressions of the performance metrics are applicable. 
To tackle the second difficulty, 
the original FBL model in~\cite{Polyanskiy_2010} has been   simplified in the following two ways: {On the one hand, a linear approximation of the FBL error probability has been proposed to facilitate the analytical derivations~\cite{Makki_linear_approx,Makki_linear_approx_2}, especially when the system involves multiple fading channels~\cite{Alcaraz_WP_linear_approx_app2,Lai_coop_NOMA,Makki_2016_WP_linear_app}} and random task arrivals~\cite{Yu_linear_approx_app1}. 
{On the other hand, instead of approximating the error expression, the channel dispersion can be simplified with a high SNR assumption~\cite{Ren_V_1_app3,She_V_1_app2,Sun_V_1_app1, Singh_2021_EE_V_1}.} 
Recently, based on the above two approaches (by either  assuming certain  SNR and blocklength regions of practical interest or simplifying the FBL model), a set of designs is provided in~\cite{Sun_joint_optimization_alt,Haghifam_alternative,Han_massive_NOMA}, e.g.,  
by optimizing with respect to multiple factors   based on the  concavity/convexity of the performance metric as a function of individual factors separately.
{{R.2.7-6}In particular, based on the convexity of the performance metric as a function of each single parameter, multi-factor joint optimization problems have been addressed via an alternating search~\cite{Sun_joint_optimization_alt,Haghifam_alternative}, successive convex approximation~\cite{He_2021_power_SCA} or integer convex optimization~\cite{Ren_V_1_app3,Han_massive_NOMA,Chao_security} (solving  all single-parameter sub-problems with different feasible integer values of blocklength)}. 

However, the aforementioned approaches face various challenges: the accuracy of a  simplified FBL model~\cite{Lancho_why_normal_approx} is often an issue, global optimum is difficult guarantee, and the computational complexity of  solution search is usually high. 
Therefore, 
it is of interest to investigate the joint convexity to characterize performance using the  original FBL model     with the normal approximation. 
In fact, although the joint convexity of the error probability approximation based on the original model can be observed (within certain regions of interest) in a set of existing works via numerical simulations~\cite{ourJSAC,Hu_numerical_convex,Avranas_numerical_convex}, to the best of our knowledge, proving this joint convexity and determining such a region is still an open challenge. 

In this paper, we study the joint convexity of error probability with respect to blocklength and transmit power based on the original FBL model in~\cite{Polyanskiy_2010} and analyze its applicability in realistic multi-user scenarios. The main contributions of this paper can be summarized as follows:
\begin{itemize} 
    \item 
    We first prove the joint convexity of the FBL error probability with respect to the transmit power and coding blocklength within a certain region of interest. 
    \item We then study the general use case of joint convexity in a setting, in which we encounter a non-convex optimization problem due to a non-convex constraint, complicating the identification of efficient joint resource allocation strategies. By exploiting a novel variable substitution method, we reformulate the optimization problem, prove its convexity, and solve it optimally with low complexity compared to the commonly applied methods.
    \item To further extend the applicability of our results, we analyze two additional practical scenarios, one with a fading channel and the other involving a cooperative relaying system. We prove that the joint convexity still holds for both cases.
    \item Via simulations, we validate our analytical results and show the significant advantages of our proposed joint optimization approaches compared to benchmark schemes, emphasizing especially the scalability to networks with massively many devices.
\end{itemize}

The remainder of the paper is organized as follows. In Section II, we present the FBL model and introduce the preliminary concepts, formulations and approximations. The joint convexity is characterized in Section III. In Section IV,  the use case of resource allocation is studied. The applicability of our results to other practical scenarios is investigated in Section V. Numerical results are presented in Section VI. Finally, we conclude the paper in Section VII. Several proofs are relegated to the Appendix. 

\vspace{-5pt}
\section{FBL Error Model and Preliminary Analysis}
\subsection{Error probability in the FBL regime}
Consider a general  wireless communication system with FBL codes, where the source transmits  a data packet containing $D$ bits using a codeword with blocklength $m$ via a wireless link to a receiver.  Assume we have perfect channel state information 
and the channel gain $z$ (including path-loss) is constant over the entire blocklength $m$.  {Then, the signal-to-noise ratio (SNR) of the received packet at the receiver is $\gamma=zp/\sigma^2$, where $p$ is the transmit power and $\sigma^2$ represents the noise spectral density. }

As the transmission is carried out in the FBL regime, i.e., the blocklength $m$ is no longer sufficiently large to be considered infinite, the assumption of communicating arbitrarily reliably at the Shannon limit is no longer accurate. In other words, the transmission can be erroneous, even if the coding rate is less than the Shannon capacity. In particular, the maximal achievable transmission rate with target error probability $\varepsilon_0$ is given approximately by~\cite{Polyanskiy_2010}
\begin{equation}
\label{eq:maximal_rate}
    r^*\approx \mathcal{C(\gamma)}-\sqrt{\frac{V(\gamma)}{m}}Q^{-1}(\varepsilon_0),
    \vspace{-5pt}
\end{equation}
where  ${\mathcal{C}} = {\log _2}( {1 + \gamma } )$ is the Shannon capacity, and ~$V$ is the channel dispersion~\cite{Chem_2015}. In particular, we have $V=1- {(1+\gamma)^{-2}}$ in a channel with fixed gain. Moreover, $Q^{-1}$ is the inverse Q-function with $Q(x)=\int^\infty_x \frac{1}{\sqrt{2\pi}}e^{-\frac{t^2}{2}}dt$.

According to~\eqref{eq:maximal_rate}, we can also approximately characterize the   (block) error probability with given transmission rate $r=\frac{D}{m}$ 
as follows:
\begin{equation}
\label{eq:error_tau}
{\textstyle
\varepsilon \!=\! {\mathcal{P}}(\gamma,r,m)\! \approx\! Q\Big( {\sqrt {\frac{m}{V(\gamma)}} ( {{\mathcal{C}}(\gamma ) \!-\! r)} {\log_e}2} \Big)\mathrm{,}}
\vspace{-5pt}
\end{equation}
{which provides a tight approximation on the error probability in the FBL regime, and has been widely applied [8]-[30]. In the remainder of the paper, we consider this expression as the FBL performance metric and characterize its convexity.}

\begin{figure}[!t]
    \centering
    \vspace{-20pt}
\includegraphics[width=0.57\textwidth,trim=0 10 0 0]{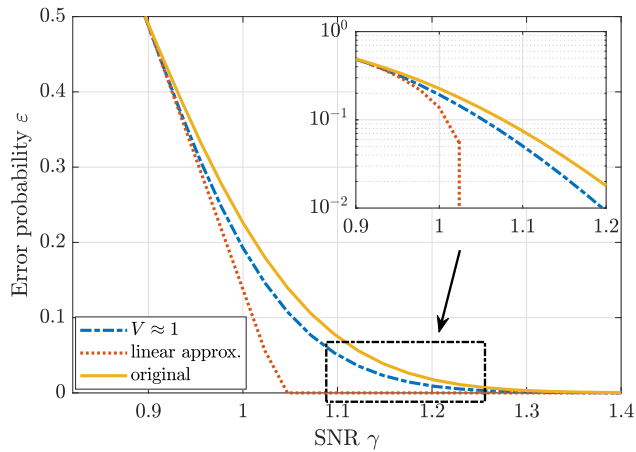}
\caption{Comparison of the error probability values computed by using the original expression in~\eqref{eq:error_tau}, linear approximation in~\eqref{eq:linear_approx}, and the approximation obtained by letting $V\approx1$. {
Note that the   curves in the zoomed-in subfigure  are convex in linear scale. The reason that they look  non-convex is due to the fact that the  subfigure is plotted in a logarithmic scale on the y-axis (in order to more clearly demonstrate the performance gap between the approximations and the original expression in a region with relatively low error probabilities).}}
\label{fig:err_comparison}
\end{figure}
\vspace*{-10pt}
\subsection{Preliminary Analysis via Additional Approximation Approaches}
Although the Equation in~\eqref{eq:error_tau} provides an accurate approximate formula for the error probability in the FBL regime, it is still analytically difficult to apply~\eqref{eq:error_tau} in efficiently addressing the related optimization problems due to its non-convexity.
In particular, $\varepsilon$ is neither a convex nor a concave function with respect to either $m$ or $\gamma$, since the Q-function itself initially behaves as a concave function for negative values of its argument and then becomes convex. To overcome this challenge, mainly two further approximation approaches have been applied in the literature:
\subsubsection{Linear approximation of $\varepsilon$} {By applying the linearization technique on the Q-function, the error probability in~\eqref{eq:error_tau} can be approximated with three line segments~\cite{Makki_linear_approx,Makki_linear_approx_2}, i.e.,}
\begin{equation}\label{eq:linear_approx}
    \varepsilon\approx
    \begin{cases}
        1,\quad \gamma <\alpha-\frac{1}{2\mu}\\
        \frac{1}{2}-\mu(\gamma-\alpha),\quad \alpha-\frac{1}{2\mu}\leq \gamma \leq \alpha+\frac{1}{2\mu}\\
        0,\quad \gamma\geq \alpha+\frac{1}{2\mu},
    \end{cases}
    \vspace{-5pt}
\end{equation}
where $\alpha=e^r-1$ and $\mu=-\sqrt{\frac{m}{2\pi(e^{2r}-1)}}$. 
{Due to the simplicity of the expression, it has been widely applied in the literature, especially when dealing with fading channels~\cite{Makki_2016_WP_linear_app, Alcaraz_WP_linear_approx_app2,Yu_linear_approx_app1}.}  

However, the major drawback of  such an approximation lies in the performance difference compared with~\eqref{eq:error_tau}. In particular, the approximation is tight when the error probability is relatively high, whereas the gap becomes significant when one enters the regime of reliable transmission with low error probabilities as shown in Fig.~\ref{fig:err_comparison}. In other words, the more reliable the transmission is, the less accurate the approximation becomes, which is counter to the purpose of studying the FBL model. For instance, enabling ultra-reliable low-latency communication (URLLC) services demands  an error probability as low as $10^{-5}$, for which the linear approximation is unable to accurately represent the actual error probability, as indicated in Fig.~\ref{fig:err_comparison}. Therefore, the applicability of such an approximation is weakened when we aim at investigating the reliability performance. 

\subsubsection{Constant approximation of $V$} 


Focusing on relatively reliable transmissions only, we can narrow the region of coding rates by assuming that $r\leq \mathcal{C}$. In particular, note that the sign of the argument of the Q-function $Q(x)$ is determined by the sign of $\sqrt{\frac{m}{V}}(\mathcal{C}-r)$. Therefore, if $r\leq \mathcal{C}$ does not hold, we have $x<0$, resulting in $\varepsilon=Q(x)>0.5$, which is too high to be of interest in practical applications requiring high reliability.
In the regime with $r\leq \mathcal{C}$, the Q-function $Q(x)$ is convex with respect to $x$. However, similar convex properties  of $\varepsilon$ with respect to other variables, e.g., $m$ or $\gamma$, are still obscure due to the complicated dependencies within \eqref{eq:error_tau}. 
To simplify the expression, the channel dispersion $V$ in the AWGN channel can be approximated as a constant instead of being expressed as a function of SNR~\cite{Sun_V_1_app1,She_V_1_app2,Ren_V_1_app3}, i.e., $V\approx 1$. 

However, it should be pointed out that the approximation is only accurate when the SNR is high and becomes loose in the  medium range of SNR as shown in Fig.~\ref{fig:err_comparison}, which is still within the region of interest. Especially when we are dealing with the error probability instead of the maximal achievable coding rate, the performance gap is  non-negligible. In fact,  $V\approx 1$ implies that the SNR $\gamma \gg 1$. 
which could be a strong assumption that oversimplifies the influence of $\gamma$~\cite{Polyanskiy_2011}.
Therefore, it is worthwhile to utilize~\eqref{eq:error_tau} directly without further approximations. Moreover, to enable the joint resource allocation approach with respect to both blocklength and power, the characterization can be more impactful if we are able to establish the joint convexity. In view of this, we analyze the convexity of $\varepsilon$ in detail in the following section.

\section{Characterization of the Convexity of Error Probability}

Recall that we are interested in the tail behavior of the system~\cite{Bennis_tail}, i.e., the transmissions are reliable with error probability $\varepsilon \leq \varepsilon_{\max}\ll 1$. In this regime, the performance cannot be guaranteed if the SNR is extremely low and/or the transmission rate is higher than the Shannon capacity. Therefore, in the considered work, we assume that $\varepsilon \leq \varepsilon_{\max}=0.1$, $\gamma\geq \gamma_{\rm th} =0$ dB and $\mathcal C(\gamma)\geq r$ always hold to facilitate the derivations. {Furthermore, we assume that the integer constraint on the blocklength can be relaxed\footnote{{We can obtain the integer solution of $m$ by comparing the possible integer neighborhood of the non-integer solution, i.e., $m_{\rm int}=\arg \min\limits_{m\in \{\lfloor m^* \rfloor, \lceil m^* \rceil\}} f(m)$, where $m^*$ is the non-integer optimal solution and $f$ is the objective function.}}, i.e., from $m \in \mathbb{Z_{+}}$ to $m \in \mathbb{R_{+}}$.} Then, we can establish the following lemma to characterize the convexity of $\varepsilon$ in the regime of interest:


\begin{theorem}
\vspace{-5pt}
\label{lemma:joint_convex}
Consider an FBL transmission with conditions $\varepsilon\leq \varepsilon_{\max}$, $\gamma\geq 1$ and $ \mathcal C(\gamma)\geq r$. The error probability  is jointly convex in 
blocklength  and transmit power, i.e., $\varepsilon$ is convex in $(m,p)$, if 
\begin{equation}
\label{eq:condition_1}
            \begin{split}
    r&>\frac{2}{\ln2(9(\gamma+1)^2-1)}\Big(-2\gamma^2-4\gamma 
    +\!\sqrt{4(\gamma^2\!+\!2\gamma)^2\!+\!\frac{2}{5}(9(\gamma\!+\!1)^2\!-\!1)\ln^2(\gamma\!+\!1)}\Big)	
        	\end{split}
        	\vspace{-5pt}
\end{equation}
or {more strictly}
\begin{equation}
\label{eq:condition_2}
    \gamma\geq \max\left\{ \frac{1}{5r\ln2},~ \frac{8}{45r^2\ln^22}\right\}.
    \vspace{-5pt}
\end{equation}
\end{theorem}
A sketch of the proof is as follows (a detailed proof is provided in Appendix A):
\begin{enumerate}
    \item 
    We first prove the convexity and monotonically decreasing property of error probability $\varepsilon$ with respect to $w=\sqrt{\frac{m}{V(\gamma)}}\big(\mathcal{C}(\gamma)-\frac{D}{m}\big)\ln{2}$.
    \item 
    Then we prove that the upper-left element in the Hessian matrix $\mathbf{H}$ of $w$ over $m$ and $\gamma$ is non-negative.
    \item {The determinant of the Hessian matrix is expressed in terms of a set of auxiliary functions, where the sign of each of them is proven to be non-negative under the condition in~\eqref{eq:condition_1}.}
    \item The condition in~\eqref{eq:condition_2} is also shown to provide an intuitive physical meaning.
\end{enumerate}
{\textbf{Remark:} \emph{{Numerically, the condition~\eqref{eq:condition_1} is fulfilled if 
$r=\frac{D}{m}\geq 0.0448   \ [\text{bits/blocklength}]$, which is true for most practical FBL applications in the region of interest. For the extreme case, a more intuitive condition~\eqref{eq:condition_2} can be applied to provide the condition between SNR $\gamma$, namely Shannon capacity $\mathcal{C}(\gamma)$, and coding rate $r$ to guarantee the convexity\footnote{In fact, (2) is only accurate if the SNR is greater than 1 and coding rate is sufficiently high~\cite{Lancho_why_normal_approx}. In the remainder of the paper, we assume that condition (4) is implicitly fulfilled.}. 
In other words, it indicates a maximal blocklength constraint with given information bits.} It is worth mentioning that Lemma~\ref{lemma:joint_convex} is also valid for the cases with the linear approximation of $\varepsilon$ and constant approximation of $V$, for which 
 we can relax the constraint~\eqref{eq:condition_1} and $\gamma >1$.}} 
 

By applying this characterization of convexity, we can now efficiently solve various joint optimization problems involving  the error probability. To illustrate this, let a $K$-th order polynomial with non-negative coefficients to be the general error probability penalty/cost function $f(\mathbf{m},\mathbf{p})$  in terms of blocklength $\mathbf{m}=\{m_1,\dots,m_N\}$ and normalized transmit power $\mathbf{p}=\{p_1,\dots,p_N\}=\{\gamma_1,\dots,\gamma_N\}$, each with $N$ variables. We further denote by $\varepsilon_i(m_i,p_i)$ the FBL error probability with respect to $m_i$ and $p_i$. Given these, we can express the penalty function $f(\mathbf{m},\mathbf{p})$ as
\begin{equation}
f(\mathbf{m},\mathbf{p})=\sum^K_{k=1} \sum^N_{i=1} a_{i,k} \varepsilon^k_i(m_i,p_i),
\vspace{-5pt}
\end{equation}
 with $a_{i,k}\geq0$ denoting the non-negative coefficients. Then, we have the following corollary to characterize the expression.
\begin{theorem1}
\vspace{-5pt}
\label{corollary:joint_convex}
The penalty function $f(\mathbf{m},\mathbf{p})$ is also jointly convex in ($\mathbf{m}$, $\mathbf{p}$) under the conditions $\varepsilon_i\leq \varepsilon_{\max}$, $\gamma_i\geq 1$ and $ \mathcal C(\gamma_i)\geq r_i$. 
\end{theorem1}
\begin{proof}
\vspace{-5pt}
For the $k$-th term $a_k\varepsilon^k_i$ of the polynomial, we have
\begin{align}
   \frac{\partial^2 (a_k\varepsilon^k_i)}{\partial \mathbf{m}^2}=\frac{\partial^2 (a_k\varepsilon^k_i)}{\partial m^2_i} \nonumber 
     =&\left\{
    \begin{array}{lr}
        a_k k(k-1)\varepsilon^{k-2}\left(\frac{\partial \varepsilon_i}{\partial m_i}\right)^2+a_kk\varepsilon^{k-1}\frac{\partial^2 \varepsilon_i}{\partial m^2_i},& \text{if}\ k\geq 2\\
        a_k\frac{\partial^2 \varepsilon_i}{\partial m^2_i},&\  \text{if}\ k=1\\
        \ 0,& \text{if}\ k=0
    \end{array}\right. \nonumber \\
    \geq & 0.
    \vspace{-5pt}
\end{align}
\vspace{-5pt}
Similarly, we can show that \vspace{-5pt}
\begin{equation}
    \frac{\partial^2 (a_k\varepsilon^k_i)}{\partial \mathbf{m}^2}\frac{\partial^2 (a_k\varepsilon^k_i)}{\partial \mathbf{p}^2}-\left(\frac{\partial^2 (a_k\varepsilon^k_i)}{\partial \mathbf{m}\partial \mathbf{p}}\right)^2\geq 0,
    \vspace{-5pt}
\end{equation}
regardless of the value of $k$. Therefore, $a_k\varepsilon^k_i$ is jointly convex in $\mathbf{m}$ and $\mathbf{p}$. As a sum of convex functions, $f(\mathbf{m},\mathbf{p})$ is also jointly convex in $\mathbf{m}$ and $\mathbf{p}$.
\end{proof}

Based on Corollary~\ref{corollary:joint_convex}, any problem that has the following form within a convex feasible regime is  convex:
    \vspace*{-10pt}
\begin{mini!}[2]
{\mathbf{m},\mathbf{p}}{f(\mathbf{m},\mathbf{p})}
{\label{problem_general}}{}
\addConstraint{\varepsilon_i}{\leq \varepsilon_{\max}, \forall i=1,\dots,N}
\addConstraint{p_{\min}}{\leq p_i \leq p_{\max}, \forall i=1,\dots,N}
\addConstraint{0}{\leq m_i \leq m_{\max},\forall i=1,\dots,N.}
\vspace{-5pt}
\end{mini!}
Clearly, as a convex problem, Problem~\eqref{problem_general} can be efficiently solved. On the other hand,  it should be pointed out that the  constraints of Problem~\eqref{problem_general}
are general and can be loose, while 
many practical scenarios have more specific  and strict constraints which are possibly  non-convex. For example, 
constraints that include coupled variables, such as total energy consumption restrictions, are non-convex.  
To demonstrate how our analytical results in Lemma~\ref{lemma:joint_convex} and Corollary~\ref{corollary:joint_convex} facilitate the joint optimal design with practical and non-convex constraints, {in the next section we analyze a general use case that is encountered frequently when dealing with multi-factor resource allocation with FBL codes.}
\vspace{-5pt}
\section{Use Case: Joint  Resource Allocation }
\label{sec:use_case}
{Consider a general multi-user scenario~\cite{Ren_V_1_app3,Sun_joint_optimization_alt}, where a set of low-latency and reliable transmissions of $N$ user equipments (UEs) are carried out via orthogonal multiple access, e.g., time-division multiple access (TDMA) or orthogonal frequency-division multiple access (OFDMA)\footnote{{We assume that the channels are frequency-flat~\cite{Ren_V_1_app3,Han_AoI}}.}.} 
 In particular, the transmission consists of $N$ slots with blocklength  $\mathbf{m}=\{m_1,..,m_N\}$, in which the server transmits packets with sizes of $\mathbf{D}=\{D_1,\dots,D_N\}$ information bits using normalized transmit power $\mathbf{p}=\{p_1,\dots,p_N\}=\{\gamma_1,\dots,\gamma_N\}$ to $N$ UEs, respectively. We define $i\in\mathcal{N}$ as the index of a user and $\mathcal{N}$ is the index set. Furthermore, the total available energy consumption is constrained by $E$ and total available  blocklength by $M$. We aim at minimizing the maximal error probability among those UEs by jointly optimizing blocklength and transmit power for each UE. To ensure the quality of transmission, the error probability of each transmission is below a threshold $\varepsilon_{\max}$ while the SNR for each user is greater than  a threshold $\gamma_{\rm th}\geq 1$. Then, the problem can be formulated as follows
     \vspace*{-10pt}
\begin{mini!}[2]
    {\mathbf{m},\mathbf{p}}{\max\{\varepsilon_{i}\}}
    {\label{problem_downlink}}{}
    \addConstraint{\sum^N_{i=1} m_i}{\leq M}
    \label{con:total_M}
    \addConstraint{\sum^N_{i=1}m_ip_i}{\leq E}
        \label{con:total_E_non_convex}
    \addConstraint{\varepsilon_i\leq\varepsilon_{\max},}{\forall i\in \mathcal{N}}
    \addConstraint{\gamma_i\geq \gamma_{\rm th},}{\forall i\in \mathcal{N}.}
    \label{con:energy_nonconvex}
    \vspace{-5pt}
\end{mini!}
 The objective function is jointly convex in~$\mathbf{m}$ and~$\mathbf{p}$ according to Lemma~\ref{lemma:joint_convex}, while the constraint~\eqref{con:total_E_non_convex} is   non-convex. 
 Hence,   problem~\eqref{problem_downlink} is non-convex in its current form. 
 In the following, we address this non-convex problem via  different strategies. 
{In particular, we first
provide an optimal solution via a variable substitution method. Meanwhile,  an alternating search and an integer search (both of which are based on the convex features of a single variable) are presented as benchmarks. It should be pointed out that these benchmarks are also the state-of-the-art approaches to provide multi-factor designs we discussed in Section I.  
At the end of the section, a comparison among all three approaches is provided to show the performance advantages of applying the characterized joint convexity in terms of   optimality, complexity and scalability.
}

\vspace{-5pt}
\subsection{Solution Based on Joint Convexity}
In particular, $\sum^N_{i=1}m_ip_i\leq E$ is a non-convex constraint, since both $m_i$ and $p_i$ are optimization variables, i.e., we  jointly optimize the blocklength and transmit power. To address this issue without compromising the constraint (e.g., by fixing one of the variables or decomposing the constraint as $m_ip_i\leq E_{\max,i}$, $\forall i\in\mathcal{N}$), we exploit the following variable substitution method: Let $m_i=\frac{1}{a_i}$ and $p_i=b^2_i$, for $i\in \mathcal{N}$. Therefore, we can reformulate problem~\eqref{problem_downlink} with new variables $\mathbf{a}$ and $\mathbf{b}$ as follows 
    \vspace*{-10pt}
\begin{mini!}[2]
{\mathbf{a},\mathbf{b}}{\max_{i\in\mathcal{N}}\{\varepsilon_{i}\}}
{\label{problem_downlink_reformulated}}{}
\addConstraint{\sum^N_{i=1}\frac{1}{a_i}}{\leq M}
\addConstraint{\sum^N_{i=1} \frac{b^2_i}{a_i}}{\leq E}
\label{con:energy_nonconvex_reformulated}
\addConstraint{\varepsilon_i\leq\varepsilon_{\max},}{\ \forall i \in \mathcal{N}}
\addConstraint{\gamma_i\geq \gamma_{\rm th},}{\ \forall i \in \mathcal{N}.}
\vspace{-5pt}
\end{mini!}
It is trivial to prove the convexity of the reformulated constraint~\eqref{con:energy_nonconvex_reformulated}, since it is the sum of $N$ convex functions. Therefore, we focus on the convexity of $\varepsilon_i$. Due to the substitution of variables, Lemma~\ref{lemma:joint_convex} should be revised, since it is no longer sufficient in proving the joint convexity. To this end, we exploit the following lemma to characterize the objective function:
\vspace*{-5pt}
\begin{theorem}
\label{lemma:joint_convex_reformulated}
The objective function of problem~\eqref{problem_downlink_reformulated} is jointly convex in $\mathbf{a}$ and $\mathbf{b}$ within the feasible set defined as $\varepsilon_i\leq\varepsilon_{\max}$ and $\gamma_i\geq \gamma_{\rm th}$, $\forall i\in\mathcal{N}$.
\end{theorem}
\vspace*{-10pt}
\begin{proof}
\vspace*{-10pt}
See Appendix B.
\end{proof}
\vspace{-5pt}
According to Lemma~\ref{lemma:joint_convex_reformulated}, problem~\eqref{problem_downlink_reformulated} is a convex problem and  can be solved efficiently by standard convex programming methods. 

Although we consider the maximal error probability as our objective function, it should be pointed out that 
Lemma~\ref{lemma:joint_convex_reformulated} can also be applied with other objectives. For instance, we can take total energy consumption, i.e., $\sum^N_{i=1}m_i p_i$,  or effective throughput, i.e., $\sum^N_{i=1} D_i(1-\varepsilon_i)$, as the objective function. However, the exact design of the metric for the resource allocation problem is out of the main scope of this paper.

{As a comparison, we also provide two other potential approaches that only leverage the individual convexity with respect to a single variable to solve the problem. Then, we discuss the advantage of the solutions based on joint convexity with respect to optimality, complexity and scalability.} 

\vspace*{-10pt}
\subsection{Approaches Based on Individual Convexity}
\subsubsection{Integer convex optimization}
{Recall that the blocklength ${m_i}$ is an integer variable and total available blocklength in the practical system is always limited. Therefore,} 
the original problem~\eqref{problem_downlink} can be decomposed into  $\frac{(M-1)!}{(M-N)!}$ sub-problems with fixed $\mathbf{m}=\mathbf{m}^o$ where $\mathbf{m}^o$ is an element in the set of possible blocklength allocation combinations $\mathcal{M}$, i.e.,
    \vspace*{-15pt}
    \begin{mini!}[2]
    {\mathbf{p}}{\max_i\{\varepsilon_{i}\}}
    {\label{problem_int}}{}
    \addConstraint{\mathbf{m}}{= \mathbf{m}^o}
    \addConstraint{\sum^N_im_ip_i}{\leq E}
    \label{con:decoposed_E}
    \addConstraint{\varepsilon_i}{\leq \varepsilon_{\max}, \gamma_i\geq \gamma_{\rm th}, \ \forall i \in \mathcal{N}.}\vspace{-15pt}
    \end{mini!}
    As proven in Lemma~\ref{lemma:joint_convex}, $\varepsilon_i$ is convex in $\mathbf{p}$ and the rest of the constraints are either affine or convex. 
    Therefore, the original problem becomes essentially an integer convex optimization problem. 
\subsubsection{Iterative search}
Similar to fixing $\mathbf{m}$, we can also decompose the original problem by fixing $\mathbf{p}=\mathbf{p}^o$, i.e.,
    \vspace*{-10pt}
        \begin{mini!}[2]
    {\mathbf{m}}{\max_i\{\varepsilon_{i}\}}
    {\label{problem_int2}}{}
    \addConstraint{\mathbf{p}}{= \mathbf{p}^o}
    \addConstraint{\sum^N_{i=1}m_i}{\leq M}
    \addConstraint{\sum^N_im_ip_i}{\leq E}
    \addConstraint{\varepsilon_i}{\leq \varepsilon_{\max}, \gamma_i\geq \gamma_{\rm th},\ \forall i \in \mathcal{N}}.
    \end{mini!}
    Since both sub-problems~\eqref{problem_int} and~\eqref{problem_int2} are convex, we can also solve the resource allocation problem in an alternating fashion. In particular, for each iteration $k$, we fix $\mathbf{p}=\mathbf{p}^{(k-1)}$ and obtain $\mathbf{m}^{(k)}$ by solving Problem~\ref{problem_int2} optimally. Then, we fix $\mathbf{m}=\mathbf{m}^{(k)}$ and obtain $\mathbf{p}^{(k)}$ by solving Problem~\ref{problem_int} optimally. We set $\mathbf{p}^{(0)}$ as the initial value for the first iteration. We repeat the above iterative procedure until convergence, and the convergence rate  is linear with computational complexity $\mathcal{O}(2N^2)$~\cite{Tseng_BCD}. 
{\vspace{-5pt}
\subsection{Comparison}
Although all three aforementioned approaches are able to provide (sub-)optimal solutions for the considered joint resource allocation problem, they are quite 
different regarding their optimality, computational complexity and scalability:
\begin{itemize}
    \item {\textbf{Joint convex optimization}: One of the major advantages of this algorithm is its ability to provide an optimal solution with the low computational complexity of $\mathcal{O}(4N^2)$ while guaranteeing global optimality. Moreover, this approach is also scalable to support massive device access in low-latency IoT networks, since the complexity only has a second-order polynomial growth rate in $N$. }
    \item \textbf{Integer convex optimization}: The globally optimal solution can be found by comparing the solutions of each decomposed sub-problem with all possible $\mathbf{m}^o$ in $\mathcal{M}$, resulting in a computational complexity of $\mathcal{O}\left(\frac{(M-1)!}{(M-N)!}N^2\right)$, which is the major drawback of this approach. Specifically, the complexity increases exponentially with the available blocklength and number of UEs. This is especially prohibitive when then number of connected UEs is high.
    \item \textbf{Iterative   search}: Recall that the complexity of iterative search is $\mathcal{O}(2N^2)$. Therefore, the advantage of iterative search is the low complexity compared to the integer convex programming. However, the solution obtained with this method is only sub-optimal and highly depends on the initialization.  To improve the performance, one can  run the algorithm multiple times with randomly chosen initial values. However, this increases the overall complexity and still cannot guarantee global optimality. Furthermore, the cost also grows further when the number of users, $N$, increases due to the larger space from which the initial value needs to be selected.
\end{itemize}

    
To elucidate the comparison}, we list the properties of all three methods in terms of optimality, complexity and scalability in Table~\ref{tab:compare}.

\begin{table*}[!t]
\centering
\vspace{-15pt}
 \begin{tabular}{||c| c| c| c|c||} 
 \hline
& Convexity & Optimality & Complexity & Scalability \\ 
 \hline
 \cellcolor{green!25} Joint convex optimization &\cellcolor{green!25} joint &\cellcolor{green!25}  globally optimal &\cellcolor{green!25}  low &\cellcolor{green!25}  yes \\ 
 \hline
 Integer convex optimization & partial &\cellcolor{green!25} globally optimal & high & no \\ 
 \hline
 Iterative search & partial & sub-optimal & \cellcolor{green!25} low & conditional \\

 \hline
\end{tabular}
\caption{{ Properties of three optimization approaches. Green highlighting indicates the best performance among the three candidates. The joint convex optimization shows significant advantages compared with other two common approaches in both theoretical perspectives, i.e., convexity and optimality, as well as practical perspectives, i.e., complexity and scalability. In particular, the joint convex optimization approach can optimally solve Problem (10) due to the joint convexity proven by Lemma 1 and 2. Therefore, the complexity is as low as $\mathcal{O}(4N^2)$. It implies that the joint convex optimization is scalable, which is critical to support massive connectivity. }}
\label{tab:compare}
\end{table*}
\vspace{-5pt}
\section{Extensions to Fading Channels and Relay Networks} 
{{In this section, we further extend and apply our   analytical results to more specific scenarios in which the formulated problems are more complicated and precluding the direct use of our analytical findings in previous sections.} 
In particular, we first investigate the joint convexity in a fading channel instead of the static channel. In this case, the error probability expression involves the fading distribution, and differs from the normal approximation in~\eqref{eq:error_tau}. Next, we consider a two-hop relaying system, where the overall error is influenced by the error of two links. Therefore, the overall error probability consists of the product of two dependent error probabilities, which increases the difficulty in establishing the convexity of the error probability. We demonstrate that our analytical results can also facilitate the analysis of joint convexity for both extensions.} 
\vspace{-10pt}
\subsection{Fading Channel}

In the previous section, joint convexity is proven under the assumption of a static channel with perfect CSI knowledge. 
{Now, let us consider a general quasi-static fading channel, the distribution of which follows central limit theorem, e.g., Rayleigh fading or Nakagami fading. 
In particular, the channel gain $z$ (including the path-loss) is constant during each transmission with coding blocklength $m$ but varies from one transmission to the next.} Then, the expression of error probability in~\eqref{eq:error_tau} is reformulated as \vspace{-5pt}
\begin{equation}
\label{eq:error_fading}
\mathbb{E}_{z}[\varepsilon]=\int^\infty_0\varepsilon f_Z(z)d z
\vspace{-5pt}
\end{equation}
where $f_Z(\cdot)$ denotes the probability density function (PDF) of the fading. 

{First, we still assume that perfect CSI is available before the considered transmission frame. Therefore, we are able to adjust the blocklength and power in each transmission.} Then, the error probability for any single transmission with known $z$ can still be evaluated by~\eqref{eq:error_tau}, unless the channel gain is lower than the threshold $z_{\rm th}$ such that Shannon capacity is lower than the transmission rate even with maximal available blocklength $M_{\max}$ and maximal transmit power $P_{\max}$, i.e., $\mathcal{C}(P_{\max} \mid z\leq z_{\rm th})<\frac{D}{M_{\max}}$. We further assume that the packets are dropped in transmissions with those parameters~\cite{She_V_1_app2}, i.e., $\varepsilon(m,p \mid z <z_{\rm th})=1$. Then, the (expected) error probability over the fading channel is given by
\begin{equation}
    \label{eq:error_fading_perfect}
    \mathbb{E}_{z}[\varepsilon]=\int^{z_{\rm th}}_0\varepsilon(m,p|z)f_Z(z)d z+\int^\infty_{z_{\rm th}}\varepsilon(m,p|z)f_Z(z)d z.
\end{equation}
{Note that the fading is quasi-static. Then, let $z(\phi)$ denote an arbitrary channel state, where $\phi=\{1,\dots,\Phi\}$. Now,  ~\eqref{eq:error_fading_perfect} can be further approximated as
\begin{equation}
\label{eq:error_fading_perfect_quan}
\mathbb{E}_{z}[\varepsilon]\approx\sum^{\Phi}_{\phi=1}\varepsilon\left(m(\phi),p(\phi)|z=z(\phi)\right)f_Z(z(\phi))\Delta \phi,
\end{equation}
where $\Delta \phi=1/\Phi$ is the channel quantization level.
As $\Phi\to \infty$, i.e., $\Delta \phi\to 0$, the approximation becomes accurate.} 
Then, we have following corollary to characterize the joint convexity:\vspace{-5pt}
\begin{theorem1}
The error probability over a quasi-static fading channel as given by $\mathbb{E}_{z}[\varepsilon]$ in~\eqref{eq:error_fading_perfect_quan} is also jointly convex in the blocklength and transmit power in each slot, i.e., as a function of $\mathbf{m}=\{m(1),\dots,m(\Phi)\}$ and  $\mathbf{p}=\{p(1),\dots,p(\Phi)\}$. \vspace{-5pt}
\end{theorem1}
\begin{proof}
Note that $\mathbb{E}_{z}[\varepsilon]$ is a weighted sum of the error probabilities under each state $z(\phi)$. In the case of $z(\phi)\leq z_{\rm th}$, we have\vspace{-5pt}
\begin{equation}
    \frac{\partial^2 \varepsilon(m(\phi),p(\phi)|z=z(\phi))f_Z(z(\phi))}{\partial (\mathbf{m},\mathbf{p})^2}=0.
    \vspace{-5pt}
\end{equation}

Moreover, in the case of $z(\phi)> z_{\rm th}$, according to Lemma 1, we have
\begin{equation}
     \frac{\partial^2 \varepsilon(m(\phi),p(\phi)|z=z(\phi))f_Z(z(\phi))}{\partial (\mathbf{m},\mathbf{p})^2}=     f_Z(z(\phi))\frac{\partial^2 \varepsilon(m(\phi),p(\phi)|z=z(\phi))}{\partial (m(\phi),p(\phi))^2}\geq 0.\vspace{-5pt}
\end{equation}
Hence, as a sum of convex functions, $\mathbb{E}_{z}[\varepsilon]$ is also convex.
\end{proof}

{However, CSI may not always available or the cost of gaining perfect CSI may be too high. In view of this, we consider a fading scenario where we only have the average CSI. Therefore, we are unable to adjust blocklength and transmit power in each transmission. In such scenarios, the instantaneous error probability may higher than 0.5 if the channel is sufficiently weak, i.e., coding rate is greater than Shannon capacity. To address this issue, following our previous work~\cite{Hu_2016_fading},  we denote $\bar z$ the medium of channel gain and $\bar z_{\rm th}$ the threshold so that the coding rate is just lower than the instantaneous Shannon capacity. For any reliable transmission with $\varepsilon\leq \varepsilon_{\max}$, i.e., $0 \leq \bar{z}_{\rm th} \leq \bar z$,  we have following corollary to characterize the convexity:}
\begin{theorem1}
\label{corollary:Fading_convex}
In a reliable transmission, the convexity of $\mathbb{E}[\varepsilon]$ with respect to blocklength $m$ or transmit power $p$ is still valid. 
\end{theorem1}
\begin{proof}
{Let $x$ presents blocklength $m$ or transmit power $p$. For the second derivative of~\eqref{eq:error_fading}, we have 
\begin{equation}
\label{eq:d2e_dx2}
	\frac{\partial ^2\mathbb{E}[\varepsilon]}{\partial x^2} 
	={\int^{{\bar{z}_{\rm th}}}_0
        \frac{\partial^2\varepsilon}{\partial x^2}f_Z(z)
        dz}
        +\underbrace{{\int^{{\bar{z}}}_{{\bar{z}_{\rm th}}}
            \frac{\partial^2\varepsilon}{\partial x^2}f_Z(z)
        dz}}_{\geq0}
        +\underbrace{\int^{\mathbf{\infty}}_{\bar{z}}
            \frac{\partial^2\varepsilon}{\partial x^2}f_Z(z)
        dz}_{\geq 0}.
\end{equation}
The last two terms are non-negative since the integral is the sum of the second derivative of error probability over $f_Z(z)$, which has been proven as convex in Lemma 1. However, the sign of the first term is still nondeterministic. }

{Note that $\bar{z}$ is the median of the channel gain. The cumulative distribution function (CDF) follows
\begin{equation}
\begin{split}
\label{eq:ineq_1}
    \int^{\infty}_{\bar{z}_{\rm th}}f_Z(z)dz&\geq \int^{\infty}_{\bar{z}}f_Z(z)dz=\int^{\bar{z}}_0f_Z(z)dz=\frac{1}{2}\geq \int^{\bar{z}_{\rm th}}_0f_Z(z)dz.
\end{split}
\end{equation}
It also holds that $\bar{z}_{\rm th}\leq \bar{z}$. Hence, we have the following inequality:
\begin{equation}
\label{eq:ineq_2}
    \int^{\mathbf{\infty}}_{{\bar{z}_{\rm th}}}
\frac{\partial^2\varepsilon}{\partial x^2}dz\geq \left|\int^{{\bar{z}_{\rm th}}}_{{0}}
\frac{\partial^2\varepsilon}{\partial x^2}dz \right|.
\end{equation}
with multiplication of both sides in~\eqref{eq:ineq_1} and~\eqref{eq:ineq_2}, the following inequality holds:
\begin{equation}
\label{eq:ineq_3}
    \int^{\mathbf{\infty}}_{{\bar{z}}}
\det \mathbf{H}(\varepsilon)f_Z(z)dz\geq \left|\int^{{\bar{z}_{\rm th}}}_{{0}}
\det \mathbf{H}(\varepsilon)f_Z(z)dz \right|.
\end{equation}
It implies that~\eqref{eq:d2e_dx2} is non-negative. as a result, the convexity of $\mathbb{E}[\varepsilon]$ with respect to $m$ or $p$ holds.} 
\end{proof}
{Based on above discussion, we can conclude that the convexity still holds for both cases with fading channel. However, they are handled differently in the optimization according to the scenarios. On one hand, if we have perfect CSI, we should optimize the blocklength and transmit power in each slot (or each $N$ slots, if we have $N$ users). On the other hand, if we only have average CSI, the best we can do is to target at the long-term error probability and optimize the blocklength and transmit power with iterative search. However, it only has to be done once for every slot. Nevertheless, all analytical results in previous sections are also applicable for quasi-static fading channel.} 
\vspace*{-10pt}
\subsection{Cooperative Relaying}
Cooperative relaying is one of the most efficient methods to mitigate fading by exploiting the spatial diversity and providing better channel quality. Thus, it is of interest to 
extend our results also to the typical two-hop relaying system. In particular, consider a relaying system, where the source (considered as node 1) transmits a data packet of $D$ bits through a relay node (considered as node 2) to the destination. Assuming that the relay employs the decode-and-forward protocol, the channel in each hop experiences block fading, and perfect CSI is available, the error probability in each hop is given by\vspace{-5pt}
\begin{equation}
    \varepsilon_i=
{\textstyle {\mathcal{P}}(p_i,\frac{D}{m_i},m_i){,}}\vspace{-10pt}
\end{equation}
where $p_i$ is the normalized transmit power of hop $i$ and $m_i$ is the blocklength. Since the packet is transmitted forward through all nodes, an error occurs if the transmission in any hop fails, i.e., the (approximated) overall error probability can be written as\vspace{-5pt}
\begin{equation}
    \varepsilon_{\rm O}=\varepsilon_1+(1-\varepsilon_1)\varepsilon_2=\varepsilon_1+\varepsilon_2-\varepsilon_1\varepsilon_2.\vspace{-10pt}
\end{equation}

We aim at minimizing the  overall error probability by jointly optimizing the blocklength and the transmit power for each hop while ensuring that the error probability of each hop is below the threshold $\varepsilon_{\max}$. Then, the optimization problem is given by\vspace*{-10pt}
\begin{mini!}[2]
{\mathbf{m},\mathbf{p}}{\varepsilon_{\rm O}}
{\label{problem_relay_series}}{}
\addConstraint{\eqref{con:total_M}}{- \eqref{con:energy_nonconvex}}\nonumber
\vspace*{-10pt}
\end{mini!}
Hence, the above optimization problem has essentially the same form of Problem~\eqref{problem_downlink} with the different objective function $\varepsilon_{\rm O}$.
{ However, Lemma~\ref{lemma:joint_convex} cannot be directly applied due to the second order term $\varepsilon_1\varepsilon_2$ in the objective function. In fact, in most existing works that study relay system with FBL codes~\cite{Pan_relay, ourJSAC, Agarwal_relay}, the term $\varepsilon_1\varepsilon_2$ is ignored. This may be inaccurate for the practical system, e.g., one of the hop has weak channel gain in a fading channel scenario. Therefore, we provide the following lemma to characterize the joint convexity of the overall error probability without ignoring $\varepsilon_1\varepsilon_2$. } \vspace*{-5pt}

\begin{theorem}
\label{lemma:relay_convex}
    The overall error probability $\varepsilon_{\rm O}$ is jointly convex in $\mathbf{m}$ and $\mathbf{p}$ within the feasible set defined as $\varepsilon_i\leq\varepsilon_{\max}$ and $\gamma_i\geq \gamma_{\rm th}$, $\forall i\in\mathcal{N}$.\vspace*{-5pt}

\end{theorem}
\begin{proof}
We first define $w_i=\sqrt{\frac{m}{V(\gamma_i)}}\big(\mathcal{C}(\gamma_i)-\frac{D}{m}\big)\ln{2}$, with which the overall error probability can be expressed as $\varepsilon_{\rm O}=Q(w_1)+Q(w_2)-Q(w_1)Q(w_2)$. Note that $\varepsilon_1=Q(w_1)$ and $\varepsilon_2=Q(w_2)$. 
Then, we first show that the overall error probability $\varepsilon_{\rm O}$ is jointly convex in $w_1$ and $w_2$. 
The Hessian matrix of $\varepsilon_{\rm O}$ with respect to $w_1$ and $w_2$ is given by\vspace*{-5pt}
\begin{equation}
\mathbf{\hat H}=\left(
\begin{array}{cc} \frac{\partial^2{\varepsilon_{\rm O}}}{\partial{w_1}^2}&\frac{\partial^2{\varepsilon_{\rm O}}}{\partial{w_1}\partial{w_2}} \\ 
\frac{\partial^2{\varepsilon_{\rm O}}}{\partial{w_2}\partial{w_1}}&\frac{\partial^2{\varepsilon_{\rm O}}}{\partial{w_2}^2}\\ 
\end{array}
\right).\vspace*{-5pt}
\end{equation}
Clearly, the upper-left element in matrix $\mathbf{\hat H}$ can be expressed as
\begin{align}
    \frac{\partial^2 \varepsilon_{\rm O}}{\partial w_1^2}=(1-Q(w_2))\frac{w_1}{\sqrt{2\pi}}e^{-\frac{w_1^2}{2}}\geq 0.
\end{align}
The above inequality holds due to the fact that $\varepsilon_2=Q(w_2)\ll 1$ and $w_1\geq 0$, which follows from $\varepsilon_1=Q(w_1)\ll 1$. Similarly, we have for the lower-right element in $\mathbf{\hat H}$,\vspace*{-5pt}
\begin{align}
    \frac{\partial^2 \varepsilon_{\rm O}}{\partial w_2^2}=(1-Q(w_1))\frac{w_2}{\sqrt{2\pi}}e^{-\frac{w_2^2}{2}}\geq 0.\vspace*{-5pt}
\end{align}

With the remaining elements formulated as\vspace*{-5pt}
\begin{equation}
    \frac{\partial^2{\varepsilon_{\rm O}}}{\partial{w_1}\partial{w_2}}=
\frac{\partial^2{\varepsilon_{\rm O}}}{\partial{w_2}\partial{w_1}}=-\frac{1}{2\pi}e^{-\frac{w_1^2}{2}-\frac{w_2^2}{2}},\vspace*{-5pt}
\end{equation}
we can obtain the determinant of matrix $\mathbf{\hat H}$ as
\vspace*{-10pt}
\begin{align}
    \det(\mathbf{\hat H})=&(1-Q(w_1))(1-Q(w_2))\frac{w_1w_2}{2\pi}e^{-\frac{w_1^2}{2}-\frac{w_2^2}{2}}
    -\frac{1}{4\pi^2}e^{-w_1^2-w_2^2}\nonumber\\
    =&\frac{1}{2\pi}e^{-\frac{w_1^2}{2}-\frac{w_2^2}{2}}\left(w_1w_2(1-\varepsilon_{\rm O})-\frac{1}{2\pi}e^{-\frac{w_1^2}{2}-\frac{w_2^2}{2}}\right)\nonumber\\
    \geq&\frac{1}{2\pi}e^{-\frac{w_1^2}{2}-\frac{w_2^2}{2}}\left(1.2^2(1-\varepsilon_{\rm O})-\frac{1}{2\pi}e^{-1.2^2}\right)
    \geq 0,
\vspace*{-10pt}
\end{align}
where the inequality holds due to the fact that $\varepsilon_i\leq \varepsilon_{\max}\ll 1$ and $\omega_1,\omega_2\geq 1.2$. Moreover, $e^{-x^2}$ is a monotonically decreasing function with respect to $x$. 
Therefore, we have proven that the overall error probability $\varepsilon_{\rm O}$ is jointly convex in $w_1$ and $w_2$. 

Note that when $w_1$ (or $w_2$) becomes larger,  we will obtain a lower error probability $\varepsilon_1$ (or $\varepsilon_2$), and the overall error probability $\varepsilon_{\rm O}$ will also be lower. In other words, the overall error probability $\varepsilon_{\rm O}$ is a decreasing function of both $w_1$ and $w_2$. On the other hand, according to Lemma~\ref{lemma:joint_convex}, it has been proven that $w_i$ is jointly concave in blocklength $m_i$ and SNR $\gamma_i$, i.e., jointly concave in $m_i$ and $p_i$. Therefore, we can conclude that the overall error probability is jointly convex in blocklength $m_i$ and power $p_i$, i.e., jointly convex in $\mathbf{m}$ and $\mathbf{p}$. 
\end{proof}

With Lemma~\ref{lemma:relay_convex}, Problem~\eqref{problem_relay_series} can also be reformulated by replacing $m_i=\frac{1}{b_i}$ and $p_i=a_i^2$ in a similar way as in Problem~\eqref{problem_downlink_reformulated} in Section~\ref{sec:use_case}. Then, it can be easily solved as a convex problem. { It is worth to mention that the joint convexity we characterized in Lemma 3 can also be applied to other scenarios, where the multiplication of error probabilities involved, such as the error probability in NOMA schemes and retransmission schemes. }
\vspace*{-10pt}
\section{Numerical Results}
In this section, we further validate our analytical characterizations via numerical results. We also show the advantage of leveraging the joint convexity compared to traditional approaches in the existing works presented as benchmarks. Unless specifically mentioned otherwise, the setup of the simulations is as follows: { We set the total available blocklength as $M=800$ [chn.use]. The normalized total energy budget as $E=2400$ [W$\cdot$ chn.use]. The packet size is set as $D_i=480$ [bits]. In such case, the bandwidth and transmission time interval is normalized to the unit value. Moreover, The path-loss is normalized to $1$ and the channels are set to experience block-fading and i.i.d. with $h_n\sim \mathcal{N}(0,1)$ while the noise spectral density $\sigma^2=0.01$. The error probability threshold is set as $\varepsilon_{\max}=0.1$.
\begin{table}[!h]
\centering
\caption{parameter setups in the numerical simulations}
\begin{tabular}{lll}
\hline
Parameter & Notation & Value\\ \hline \hline
    Total available blocklength      & $M$         & 800 [chn.use] \\ \hline
    (Normalized) total energy budget      &     $E$     & 2400 [W$\cdot$chn.use] \\ \hline
    Channel of user $i$     &      $h_i$    & $\sim\mathcal{N}(0,1)$ \\ \hline
    Noise spectral density & $\sigma^2$ & $0.01$ \\ \hline
    Error probability threshold      &      $\varepsilon_{\max}$    & 0.1 \\ \hline
    SNR threshold      &      $\gamma_{\text{th}}$    & 1 \\ \hline
\end{tabular}
\end{table}
}

We first illustrate the feasible set according to Lemma~\ref{lemma:joint_convex} in Fig.~\ref{fig:feasible_set_single}. We plot the heat map with respect to  transmission rate $r$ and SNR $\gamma$. The color indicates the value of the decoding error probability at that point, while the regimes without color or with transparent color are infeasible. We also plot the zoomed-in view of the boundary of condition~\eqref{eq:condition_1}. As discussed in previous sections, the area outside of the boundary is tiny and is not of interest considering the practical scenarios due to the extremely low transmission rates $r<0.04$. On the other hand, the upper white area is the regime where $\mathcal{C}\leq r$, which also belongs to the non-convex set as we discussed in Section II, where the Q-function is always greater 0.5, i.e., $Q(x)\geq 0.5,\ \forall x\leq 0$. Furthermore, recall that we consider the reliable transmission with a reasonable SNR threshold $\gamma\geq 1$. Therefore, the feasible regime is suitable for most practical application scenarios.
\begin{figure}[!t]
    \centering
\includegraphics[width=0.57\textwidth,trim=0 10 0 0]{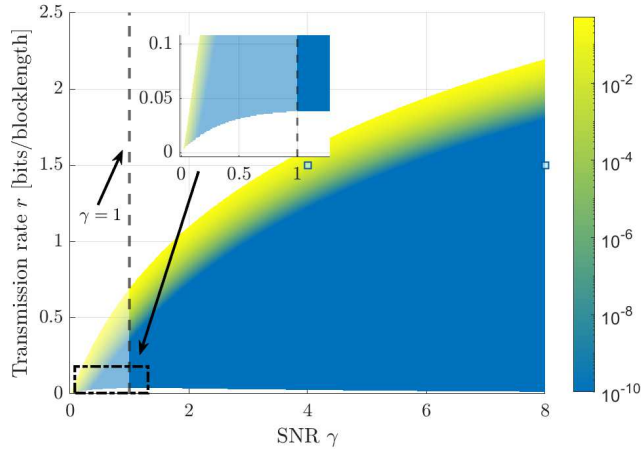}
\caption{Feasible set of the joint convexity for error probability in FBL regime. The color represents the value of $\varepsilon(\gamma,D,m)$. The zoom-in figure shows the boundary of condition~\eqref{eq:condition_1}.}
\label{fig:feasible_set_single}
\end{figure}

Next, we investigate the impact of transmit power $p$ and blocklength $m$ by showing the error probability against one of the variables while varying another, respectively. To present the influence of the channel model, we also plot $\varepsilon$ for both static channel and (slow) fading channel in Fig.~\ref{fig:convexity_single}. Specifically, we analyze the error probability in~\eqref{eq:error_tau} in terms of instantaneous SNR in the static channels and error probability in~\eqref{eq:error_fading} in terms of average SNR in fading channels. Overall, we can observe that the curves are initially concave then convex in both $p$ and $m$, taking the same shape as the Q-function regardless of whether static or fading channels are considered. This observation confirms our analytical results in Corollary~2.


\begin{figure}[!t]
    \centering
\includegraphics[width=0.57\textwidth,trim=0 10 0 0]{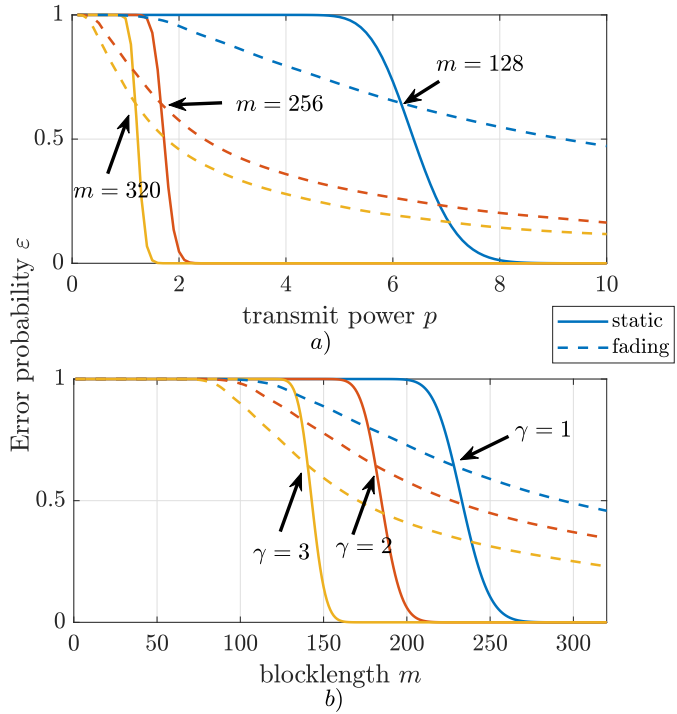}
\caption{Decoding error probability $\varepsilon$ against $a)$ transmit power $p$ ; $b)$ blocklength $m$ in various setups. Both static channel with instantaneous SNR (solid lines) and fading channel with average SNR (dash lines) are plotted.}
\label{fig:convexity_single}
\end{figure}

However, above figures only investigate the convexity of the error probability with a single user. To further provide insight on the joint convexity in practical scenarios with two users, we plot in Fig.~\ref{fig:feasible_set} the maximal error probability $\max_i \varepsilon_i$ against blocklength $m_1$ and transmit power $p_1$ of user 1, as well as the new set of variables $a_1$ and $b_1$. As a comparison, both the 3D plot of $\max_i \varepsilon_i$ and the corresponding heat map are provided. It should be pointed out that the visible area of the figure represents the feasible set of Problem~\eqref{problem_downlink}. As discussed in Section IV, $\varepsilon_i$ is jointly convex in $p_i$ and $m_i$, which is confirmed by the 3D plot in Fig.~\ref{fig:2}. Therefore, the constraint $\varepsilon_i\leq \varepsilon_{\max}$ is also convex, with which the feasible set is bounded by the bottom line in Fig.~\ref{fig:1}. However, the energy constraint $m_1p_1+m_2p_2\leq E$ is non-convex due to the multiplication of two variables, which results in a non-convex boundary depicted by the upper line in Fig.~\ref{fig:1}. To tackle this issue, we let $a_i=\frac{1}{m_i}$ and $b_i=\sqrt{p_i}$. Then, the energy constraint becomes convex, while the convexity of $\varepsilon_i$ is maintained. As a result, we have a jointly convex objective function (plotted in Fig.~\ref{fig:4}) with a convex feasible set (shown as the intersection of two boundaries in Fig.~\ref{fig:3}). These phenomena confirm our analytical results in Lemma~\ref{lemma:joint_convex_reformulated}. 

\begin{figure}[!t]
    \centering 
\begin{subfigure}{0.3\textwidth}
  \includegraphics[width=\linewidth]{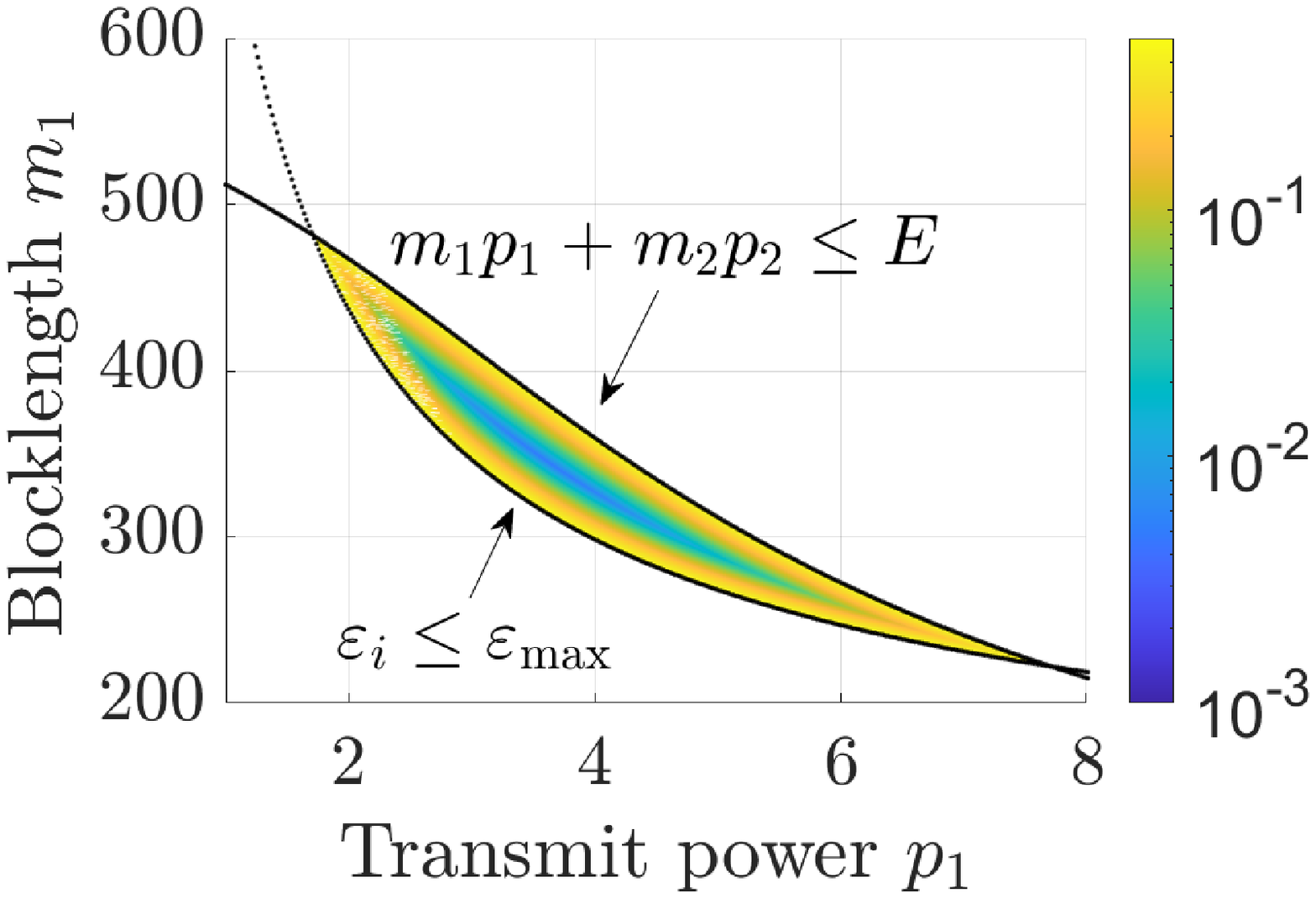}
  \caption{heat map of error probability with variable $m_1$ and $p_1$.}
  \label{fig:1}
\end{subfigure} 
\begin{subfigure}{0.3\textwidth}
  \includegraphics[width=\linewidth]{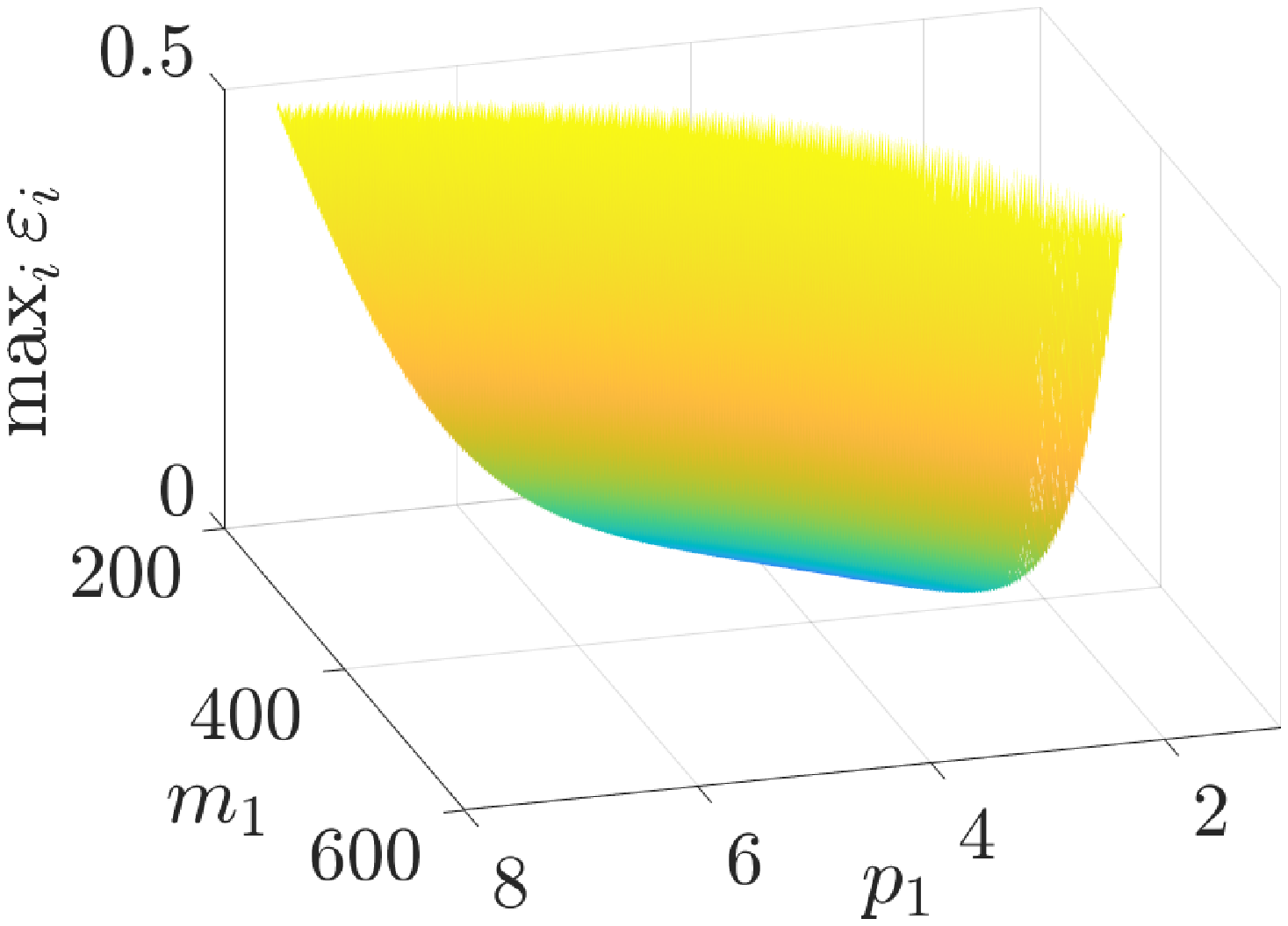}
  \caption{3D plot for error probability with variable $m_1$ and $p_1$.}
  \label{fig:2}
\end{subfigure}

\bigskip
\begin{subfigure}{0.3\textwidth}
  \includegraphics[width=\linewidth]{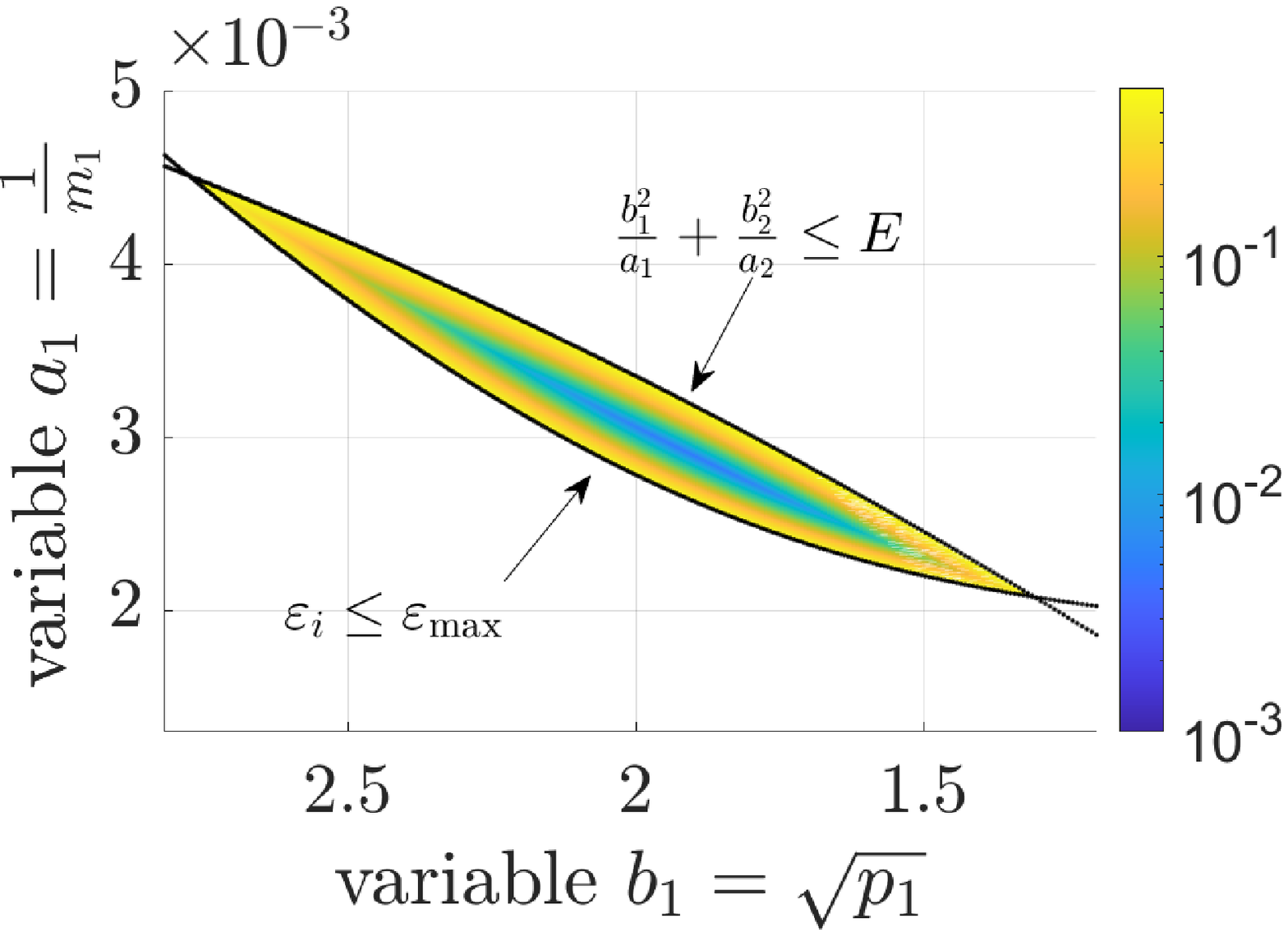}
  \caption{heat map of error probability with variable $b_1$ and $a_1$.}
  \label{fig:3}
\end{subfigure} 
\begin{subfigure}{0.3\textwidth}
  \includegraphics[width=\linewidth]{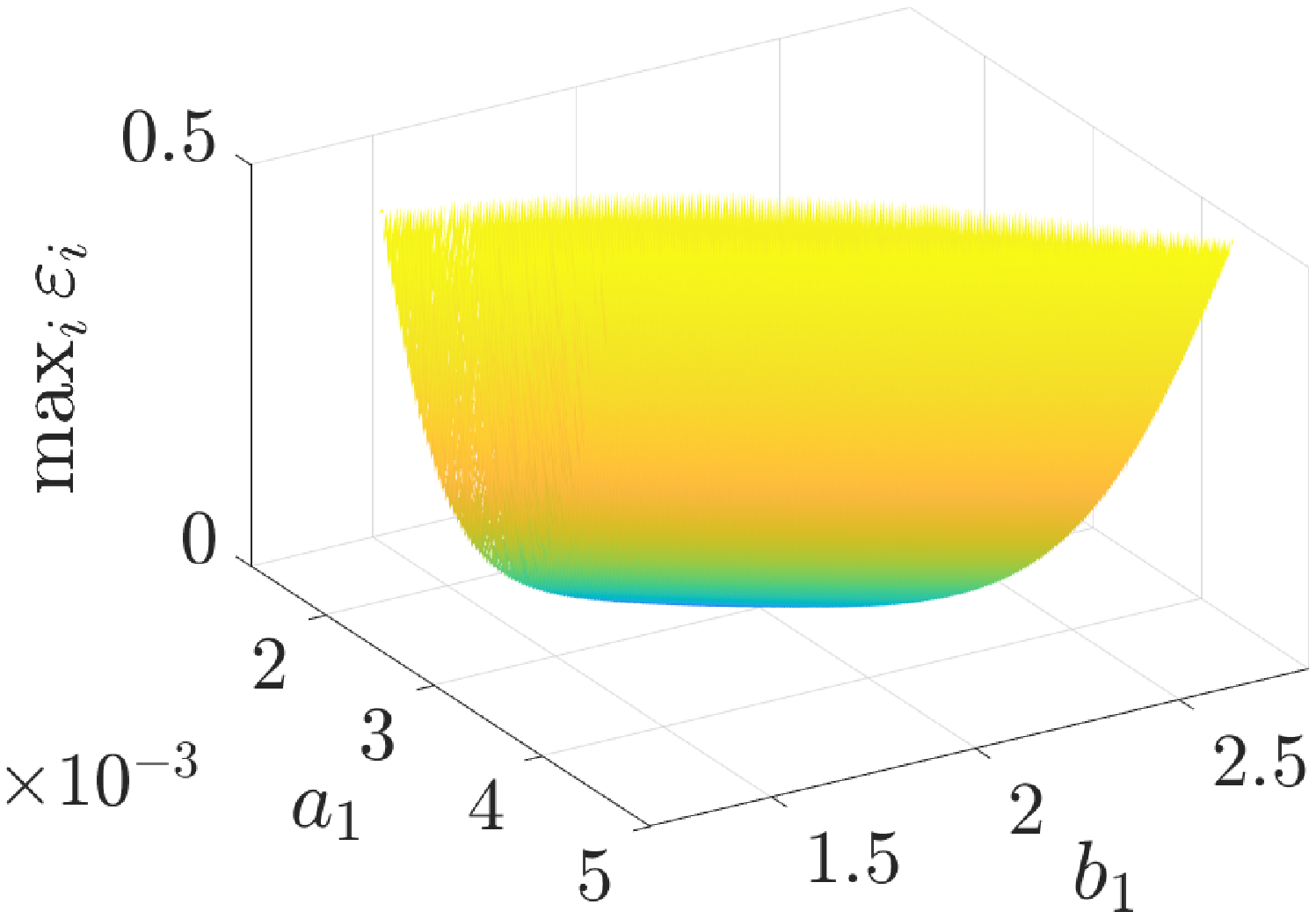}
  \caption{3D plot for error probability with variable $b_1$ and $a_1$.}
  \label{fig:4}
\end{subfigure}\hfill 
\caption{maximal error probability $\max_i\  \varepsilon_i$ with constraint~\eqref{con:energy_nonconvex} and~\eqref{con:energy_nonconvex_reformulated} against corresponding $m_1$ and $p_1$, as well as $b_1=\frac{1}{p_1}$ and $a_1=\sqrt{m_1}$.}
\label{fig:feasible_set}
\end{figure}

Furthermore, we show the advantage of our proposed joint optimization solution by comparing the results with integer search and alternating search. In particular, we plot the total error probability against total available blocklength $M$ with various total energy constraints in Fig.~\ref{fig:performance_comparison}. Overall, the error probability decreases exponentially when we increase the total available blocklength $M$ for all curves. However, we can observe significant performance gap between different approaches. On one hand, the proposed joint optimization design outperforms the alternating search method. It reveals the major disadvantage of alternating search. Specifically, the alternating search only offers sub-optimal solutions while our design is able to guarantee the global optimum. Moreover, longer total available blocklength $M$ or smaller energy budget, i.e., lower power per blocklength, also enlarges the difference. This implies that alternating search may provide an acceptable solution if the interplay of those two variables has only a weak influence on the system and can be approximately decoupled, e.g., the channels of UEs are homogeneous. On the other hand, we observe that the performance of integer search (almost) matches the performance of joint optimization. In fact, both approaches can achieve the globally optimal solution. 
{However, it should be emphasized that, although the integer search can provide globally optimal solution theoretically, it becomes computationally infeasible in practice for large IoT network scenarios, where massive number of sensor devices is involved due to the low computational efficiency.}  

\begin{figure}[!t]
    \centering
\includegraphics[width=0.57\textwidth,trim=0 10 0 0]{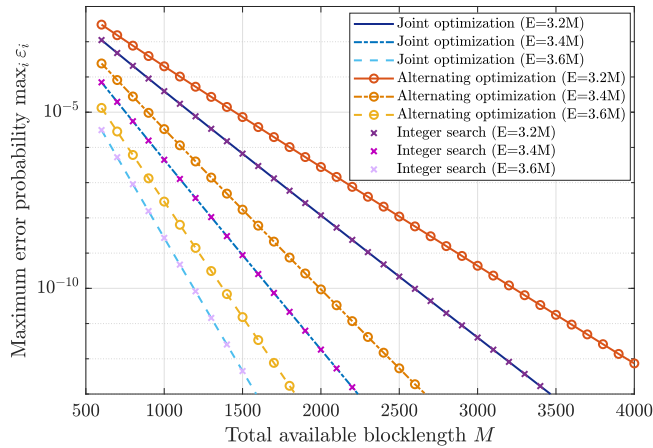}
\caption{Total error probability versus total available blocklength $M$. Performance comparison between our solution obtained with joint convexity feature, integer search and alternating search in various setups.}
\label{fig:performance_comparison}
\end{figure}

Another important parameter influencing the system performance is the transmitted data size~$\mathbf{D}$, which is investigated in Fig.~\ref{fig:D_1}. In particular, we plot the maximum error probability against average data size $\overline{D}$ with 5 users. The data size is heterogeneous and follows a fixed ratio, i.e., $\mathbf{D}=\{0.8\overline{D},0.9\overline{D},\overline{D},1.1\overline{D},1.2\overline{D}\}$ and $\sum^5_{i=1} D_i=5\overline{D}$. We also illustrate the performance difference between all three approaches in setups with varying total available blocklength $M=\{1000, 1250, 1500\}$. Clearly, increasing the average data size $\overline{D}$ significantly increases the error probability. Moreover, we can also observe that there exists clearly a trade-off between the data size limitation that the system is able to support and the total available resources (available blocklength in this case). For example, for a target error probability of $10^{-4}$, only $\overline{D}=355$ bits can be transmitted with $M=1000$. Meanwhile, we can transmit up to $\overline{D}=540$ bits with $M=1500$ at the same error level. From Fig~\ref{fig:D_1}, we also obtain similar observations as in Fig~\ref{fig:performance_comparison}, i.e., joint optimization and integer search are able to provide  better results than alternating search. Interestingly, the influence of $\overline{D}$ on the performance gain between the global optimum (obtained via joint optimization and integer search) and the local optimum (obtained via alternating search)  is insignificant. Therefore, the applicability of those approaches remains the same regardless of~ $\overline{D}$. 
\begin{figure}[!t]
    \centering
\includegraphics[width=0.57\textwidth,trim=0 10 0 0]{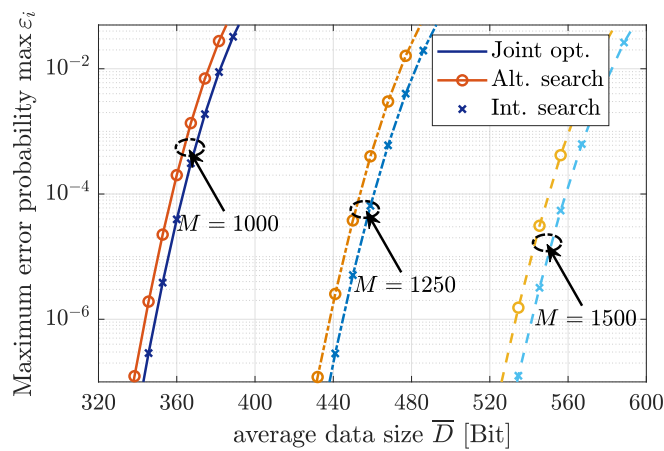}
\caption{Maximum error probability versus average data size $\overline{D}$ with various total available blocklength $M$. Data size for each user is set as $\mathbf{D}=\{0.8\overline{D},0.9\overline{D},\overline{D},1.1\overline{D},1.2\overline{D}\}$.}
\label{fig:D_1}
\end{figure}

Finally, we evaluate the impact of the channel gain on the performance. Specifically, simulations are carried out considering a two-hop relaying system, where the overall error probability $\varepsilon_{\rm O}$ is determined by the combination of error probability in each hop $\varepsilon_1(m_1,p_1)$ and $\varepsilon_2(m_2,p_2)$. Fig~\ref{fig:D_2} illustrates the optimal overall error probability $\varepsilon^*_{\rm O}$ and optimal resource allocation ratio \{$\frac{m_1}{m_2}$,  $\frac{p_1}{p_2}$\} versus the channel gain of the second hop $h_2$, i.e., the gain between the relay and destination. Meanwhile, we set the channel gain between the source and relay as constant, i.e., $h_1=1$. The optimal solutions are obtained via our proposed joint optimization approach, i.e., the solutions are globally optimal. As expected, increasing the channel gain results in the decrease of $\varepsilon^*_{\rm O}$. Since $\varepsilon_{\rm O}$  is a combination of $\varepsilon_1$ and $\varepsilon_2$, the optimal resource allocation scheme is to distribute the blocklength and transmit power so that the performance is balanced in both hops. For instance, when $h_1=h_2=1$, we have $m_1=m_2$ and $p_1=p_2$. This provides us insight on resource allocation in a homogeneous system. Specifically, when the users are considered as homogeneous (or if no information is available), an effective approach is to uniformly allocate the radio resources to each user.

\begin{figure}[!t]
    \centering
\includegraphics[width=0.57\textwidth,trim=0 10 0 0]{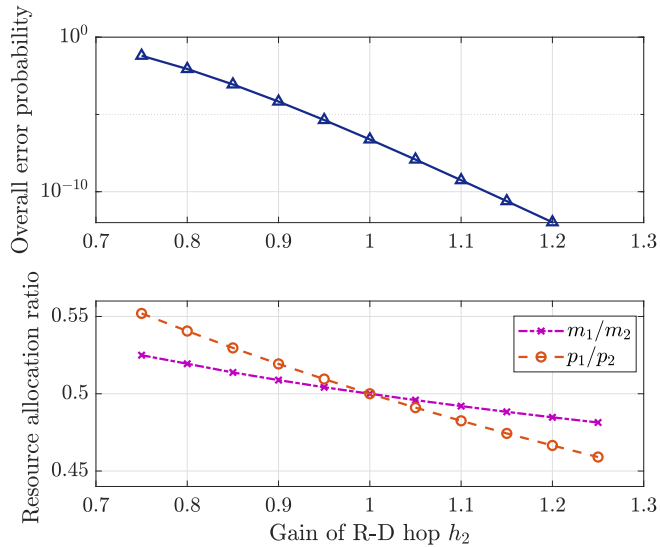}
\caption{Optimal overall error probability $\varepsilon^*_{\rm O}$ and optimal resource allocation ratio \{$\frac{m_1}{m_2}$,  $\frac{p_1}{p_2}$\} against channel gain of the second hop $h_2$ under the two hop relaying system.}
\label{fig:D_2}
\end{figure}
\vspace*{-10pt}
\section{Conclusion}
In this work, we have investigated the fundamental characteristics of the FBL error probability to enable joint optimization designs for reliable transmissions. We have proved the joint convexity of the error probability with respect to transmit power and blocklength within a certain region of practical interest. We have also studied a general use case with a non-convex joint optimization problem in a multi-user scenario. By exploiting the variable substitution method, we have reformulated the problem as a convex problem and discussed the advantages of such an approach compared to other commonly applied methods. To further extend the applicability of this work, we have considered the joint convexity in two more practical scenarios, i.e., in fading channels and in cooperative relaying. In particular, we proved that the joint convexity still holds for these two scenarios. Via simulations, we have validated our analytical results and demonstrated the performance gain of our proposed approaches compared to benchmarks.

\appendices
\section{Proof of Lemma~\ref{lemma:joint_convex}}
To facilitate the analysis, we first let $w=\sqrt{\frac{m}{V(\gamma)}}\big(\mathcal{C}(\gamma)-\frac{D}{m}\big)\ln{2}$. As a result, the error probability can be rewritten as a composite function $\varepsilon=\varepsilon(w(\gamma,m))$. 
By taking the first and second derivatives of the error probability $\varepsilon$ with respect to $w$, we have
\vspace*{-5pt}
\begin{equation}
\label{eq:de_dw}
\frac{\partial{\varepsilon}}{\partial{w}}=\frac{\partial \left(\int^\infty_w \frac{1}{\sqrt{2\pi}}e^{-\frac{t^2}{2}}dt \right)}{\partial w}=  -\frac{1}{\sqrt{2\pi}}e^{-\frac{w^2}{2}}< 0,
\vspace*{-10pt}
\end{equation}
\vspace*{-5pt}
\begin{equation}
\frac{\partial^2{\varepsilon}}{\partial{w^2}}=\frac{w}{\sqrt{2\pi}}e^{-\frac{w^2}{2}}\geq 0.
\vspace*{-10pt}
\end{equation}
The above inequalities hold due to the fact that $w$ is always non-negative in the considered high-reliability scenario where $\mathcal C(\gamma) \geq r$. Clearly, the error probability $\varepsilon$ is convex and monotonically decreasing in $w$ within the considered regime, i.e., $w\geq 0$  ($\mathcal C (\gamma)\geq r$). Therefore, the joint convexity of $\varepsilon$ in $\gamma$ and $m$ can be proved if we can prove that the auxiliary variable $w$ is jointly concave with respect to $\gamma$ and $m$. 

Next, to prove the joint concavity of $w$ in $m$ and $\gamma$, we show that the Hessian matrix $\mathbf{H}$ of auxiliary variable $w$ is negative semi-definite, where 
the Hessian matrix  $\mathbf{H}$ can be expressed as  
\vspace*{-5pt}
\begin{equation}
\label{eq:hessian}
\mathbf{H}=\left(
\begin{array}{cc} \frac{\partial^2{w}}{\partial{m}^2}&\frac{\partial^2{w}}{\partial{m}\partial{\gamma}} \\ 
\frac{\partial^2{w}}{\partial{\gamma}\partial{m}}&\frac{\partial^2{w}}{\partial{\gamma}^2}\\ 
\end{array}
\right).
\vspace*{-10pt}
\end{equation}

Subsequently, we investigate each component in matrix $\mathbf{H}$. 
On the one hand, for the partial derivatives of $w$ with respect to  $m$, we have
\vspace*{-5pt}
\begin{align}
\frac{\partial{w}}{\partial{m}}&=\frac{1}{2}m^{-\frac{1}{2}}V^{-\frac{1}{2}}\mathcal{C}\ln{2}+\frac{1}{2}m^{-\frac{3}{2}}V^{-\frac{1}{2}}D\ln{2}\geq 0, \\
\frac{\partial^2{w}}{\partial{m}^2}&=\!-\frac{1}{4}m^{-\frac{3}{2}}V^{-\frac{1}{2}}\mathcal{C}\ln{2}\!-\!\frac{3}{4}m^{-\frac{5}{2}}V^{-\frac{1}{2}}D\ln{2}\leq 0,\!\!
\vspace*{-10pt}
\end{align}
and therefore the term $w$ is concave in blocklength~$m$. 

On the other hand, the partial derivative of $w$ with respect to SNR $\gamma$ is given by 
\vspace*{-5pt}
\begin{equation}
\begin{split}
\frac{\partial{w}}{\partial{\gamma}}=&\frac{m^{\frac{1}{2}}V^{-\frac{1}{2}}}{(\gamma^2+2\gamma)(\gamma+1)}
\underbrace{\left(\gamma^2+2\gamma-\ln(\gamma+1)\right)}_{\Delta_1}
+\frac{1}{2}m^{-\frac{1}{2}}V^{-\frac{3}{2}}D\ln{2}\frac{2}{(1+\gamma)^3},
\end{split}
\vspace*{-10pt}
\end{equation}
where $\Delta_1$ is a function of SNR $\gamma$. In addition, the derivative of $\Delta_1$ with respect to $\gamma$ is given by $\frac{\partial\Delta_1}{\partial\gamma}=2\gamma+2-\frac{1}{\gamma+1}$, which is monotonically increasing in $\gamma\geq 1$, so that we have $\frac{\partial\Delta_1}{\partial\gamma}\geq3.5>0$ when $\gamma\geq 1$. Hence, we can also obtain that $\Delta_1\geq \Delta_1|_{\gamma=1}=3-\ln2>0$, i.e., $w$ is increasing in $\gamma$.
Furthermore, we can formulate the second-order derivative of $w$ with respect to $\gamma$ as 
\begin{align}
\frac{\partial^2{w}}{\partial{\gamma}^2}=&\frac{\sqrt{m}}{(\gamma(\gamma+2))^{\frac{5}{2}}}\left(-(\gamma+1)^3+\frac{1}{\gamma+1}\right. 
\left.+3\ln2(\gamma+1)\Big(\log_2(\gamma+1)-\frac{D}{m}\Big)\right).
\end{align}
To investigate the sign of $\frac{\partial^2{w}}{\partial{\gamma}^2}$, we define the function 
\begin{equation}
\Delta_2(\gamma)=-(\gamma+1)^3+\frac{1}{\gamma+1}+3\ln2(\gamma+1)\log_2(\gamma+1).
\end{equation}
The corresponding first-order derivative of $\Delta_2(\gamma)$ is given by 
\begin{align}
    \frac{\partial\Delta_2(\gamma)}{\partial\gamma}&=-3(\gamma+1)^2-\frac{1}{(\gamma+1)^2}+3+3\ln(\gamma+1) \nonumber 
    \leq -3(\gamma+1)^2-\frac{1}{(\gamma+1)^2}+3+3\gamma\nonumber \\
    &=-3(\gamma+1)\gamma-\frac{1}{(\gamma+1)^2}\leq 0.
\end{align}
Therefore, we can obtain that $\Delta_2(\gamma)$ is decreasing in $\gamma\geq 1$ and $\Delta_2(\gamma)\leq \Delta_2(1)=6\ln2+0.5-8<0$, $\forall \gamma\geq 1$. As a result, we have
\begin{equation}
    \frac{\partial^2{w}}{\partial{\gamma}^2}\!=\!\frac{\sqrt{m}}{(\gamma(\gamma\!+\!2))^{\frac{5}{2}}}\left(\!\Delta_2(\gamma)\!-\!3\ln2(\gamma\!+\!1)\frac{D}{m}\right)\!\leq \!0, 
\end{equation}
so that the term $w$ is also concave in SNR $\gamma$. 

So far, we have proved that both diagonal elements of the Hessian matrix $\mathbf H$ defined in  \eqref{eq:hessian} are non-positive. To prove the joint concavity of $w$ in $m$ and $\gamma$, the final step is to guarantee a non-negative determinant for matrix $\mathbf H$. We first derive the expression for the remaining elements in the matrix $\mathbf H$, i.e.,
\begin{equation}
\!\frac{\partial^2{w}}{\partial{m}\partial{\gamma}}\!=\!\frac{\partial^2{w}}{\partial{\gamma}\partial{m}}\!=\!\frac{m^{-\frac{1}{2}}V^{-\frac{1}{2}}\ln{2}}{2(\gamma+1)}\!\!\left(\!-{\frac{\mathcal{C}(\gamma)+r}{\gamma^2+2\gamma}\!+\!\frac{1}{\ln2}}\!\right)\!\!.\!\!\!\!\!
\end{equation}
Then, we formulate the determinant of $\mathbf H$, i.e.,
\begin{equation}
\det(\mathbf{H})=\frac{\partial^2{w}}{\partial{m}^2}\frac{\partial^2{w}}{\partial{\gamma}^2}
    -\left(\frac{\partial^2{w}}{\partial{m}\partial{\gamma}}
        \right)^2,
\end{equation}
which can be expressed in more detail as in~\eqref{eq:y_4}. 
\begin{figure*}   
\vspace*{-10pt}
\begin{align}
\det(\mathbf H)=&\frac{m^{-1}\ln{2}}{(\gamma^2\!+\!2\gamma)}\!\left(\!\underbrace{\frac{\mathcal{C}(\gamma)}{(\gamma^2\!+\!2\gamma)\ln2}\!-\!\frac{3\mathcal{C}(\gamma)^2}{4(\gamma^2\!+\!2\gamma)}\!+\!\frac{\mathcal{C}(\gamma)}{4\ln2}\!-\!\frac{1}{4(\ln2)^2}\!-\!\frac{3\mathcal{C}(\gamma)^2}{5(\gamma^2\!+\!2\gamma)^2}}_{\Delta_3}\right.\nonumber \!\\
\vspace*{-10pt}
&\qquad \qquad\left.+\!\underbrace{\frac{3r}{4\ln2}\!-\!\frac{3r\mathcal{C}(\gamma)}{2(\gamma^2\!+\!2\gamma)}-\frac{2r\mathcal{C}(\gamma)}{(\gamma^2\!+\!2\gamma)^2}}_{\Delta_4}\right.\nonumber\\
&\qquad \qquad \qquad \left. +\underbrace{\frac{2r}{(\gamma^2\!+\!2\gamma)\ln2}\!+\!\frac{9r^2}{4(\gamma^2\!+\!2\gamma)}\!+\!\frac{2r^2}{(\gamma^2\!+\!2\gamma)^2}\!-\!\frac{2\mathcal{C}(\gamma)^2}{5(\gamma^2\!+\!2\gamma)^2}}_{\Delta_5}\right).\label{eq:y_4} 
\end{align}
\hrule
\end{figure*}
More specifically, we have defined three terms in \eqref{eq:y_4} as $\Delta_3$, $\Delta_4$ and $\Delta_5$, respectively. Clearly, we have $\det(\mathbf H)$ non-negative if we prove that all the three terms, i.e., $\Delta_3$, $\Delta_4$ and $\Delta_5$, are non-negative. 
In the following, we will successively investigate the signs of these three terms. 

Note that the term $\Delta_3=\Delta_3(\gamma)$ is a pure function of SNR $\gamma$, while $\mathcal C=\log_2(1+\gamma)$. In other words, there are no additional variables in the function $\Delta_3(\gamma)$, which motivates us to numerically determine the sign of function $\Delta_3(\gamma)$ when $\gamma\geq 1$. 
By studying $\frac{\partial \Delta_3}{\partial \gamma}$, we can obtain that within the region of $\gamma\in[1,\infty)$, 
 the function $\Delta_3(\gamma)$ is monotonic. 
Besides, we also have $\Delta_3(1)\approx 0.0046>0$ and $\lim_{\gamma\to\infty}\Delta_3(\gamma)=+\infty$, which indicates that $\Delta_3(\gamma)\geq 0$, $\forall \gamma\geq 1$. 

As for the term $\Delta_4$, 
we can show that 
\vspace*{-5pt}
\begin{equation}
\begin{split}
\Delta_4&=\frac{3r}{4\ln2}-\frac{3r\mathcal{C}(\gamma)}{2(\gamma^2+2\gamma)}-\frac{2r\mathcal{C}(\gamma)}{(\gamma^2+2\gamma)^2}
=\frac{r}{\ln2}\left(\frac{3}{4}-\frac{3\ln(\gamma+1)}{2(\gamma^2+2\gamma)}-\frac{2\ln(\gamma+1)}{(\gamma^2+2\gamma)^2}\right)\\
&\overset{\gamma\geq 1}{\geq} \frac{r}{\ln2}\left(\frac{3}{4}-\frac{3\ln{2}}{6}-\frac{2\ln{2}}{9}\right)\geq 0.
\end{split} 
\vspace*{-10pt}
\end{equation}
The above inequalities hold due to the fact that both functions $\frac{3\ln(\gamma+1)}{2(\gamma^2+2\gamma)}$ and $\frac{2\ln(\gamma+1)}{(\gamma^2+2\gamma)^2}$ are monotonically decreasing with respect to $\gamma\geq 1$. 

Finally, for $\Delta_5$, we have
\vspace*{-5pt}
\begin{equation}\label{eq:con_condi}
    \begin{split}
    &\Delta_5>0\\
    \!\!\!\Leftrightarrow~ &r\geq\frac{2}{\ln2(9(\gamma+1)^2-1)}\Big(-2\gamma^2-4\gamma 
    +\!\sqrt{4(\gamma^2\!+\!2\gamma)^2\!+\!\frac{2}{5}(9(\gamma\!+\!1)^2\!-\!1)\ln^2(\gamma\!+\!1)}\Big)\\
    &~\triangleq \Delta_6(\gamma).
    \end{split}
\vspace*{-10pt}
\end{equation}
Note that although the above condition seems complicated, we can still investigate the function $\Delta_6(\gamma)$ for an insight on the condition. Since $\Delta_6(\gamma)$ is a pure function of $\gamma$, we can also perform a numerical evaluation. Via letting $\frac{\partial \Delta_6}{\partial \gamma}=0$, we can obtain the unique root $\gamma_2^*\approx 1.2408$ within the regime of $\gamma\geq 1$. And accordingly, we have that $\Delta_6(\gamma)$ is  
non-decreasing when $\gamma\in[\gamma_2^*,+\infty)$, while it holds $\Delta_6(\gamma)\leq\Delta_6(\gamma_2^*)\approx 0.0448$. In other words, a minimum coding rate constraint of $r\geq 0.0448$ [bits/blocklength], will be sufficient for this condition, which is also true in most practical scenarios. As a result, we have that under the condition \eqref{eq:con_condi}, the joint concavity of $w$ with respect to $m$ and $\gamma$ holds. Namely, in the considered high-reliability scenario, the error probability $\varepsilon$ is jointly convex in blocklength $m$ and SNR $\gamma$ under the condition \eqref{eq:con_condi}.

In addition, 
since the condition \eqref{eq:con_condi} lacks an intuitive physical meaning, we further derive a simple, but tight bound for the joint convexity. By comparing all three positive terms in $\Delta_5$, we can observe the following:
\begin{itemize}
    \item When $r\leq \frac{8}{9\ln2}$, $\frac{2r}{(\gamma^2+2\gamma)\ln2}$ is larger than $\frac{9r^2}{4(\gamma^2+2\gamma)}$. If it holds that $\frac{2r}{(\gamma^2+2\gamma)\ln2}-\frac{2\mathcal{C}(\gamma)^2}{5(\gamma^2+2\gamma)^2}\geq 0$, then we will have $\Delta_5\geq 0$, i.e., the joint convexity of $\varepsilon$ in $m$ and $\gamma$ holds. More specifically, given $r\leq \frac{8}{9\ln2}$, we have the constraint for $\gamma$ guaranteeing the convexity as
\vspace*{-5pt}
\begin{equation}
\begin{split}
&\frac{2r}{(\gamma^2+2\gamma)\ln2}-\frac{2\mathcal{C}(\gamma)^2}{5(\gamma^2+2\gamma)^2}
=\frac{1}{\big(\gamma^2+2\gamma)\ln2}(2r-\frac{2\ln^2(\gamma+1)}{5(\gamma^2+2\gamma)\ln2}\big)\geq0\\
\Leftarrow&2r\geq\frac{2(\gamma+1)}{5(\gamma^2+2\gamma)\ln2}
\Leftarrow
\gamma\geq\frac{1}{5r\ln2}.
\end{split}
\vspace*{-10pt}
\end{equation}
\item When $r>\frac{8}{9\ln2}$, we have  $\frac{2r}{(\gamma^2+2\gamma)\ln2}<\frac{9r^2}{4(\gamma^2+2\gamma)}$. Therefore, $\Delta_5\geq 0$ holds when we have $\frac{9r^2}{4(\gamma^2+2\gamma)}-\frac{2\mathcal{C}(\gamma)^2}{5(\gamma^2+2\gamma)^2}\geq 0$, which results in the following constraint on $\gamma$
\vspace*{-5pt}:
\begin{equation}
\begin{split}
&\frac{9r^2}{4(\gamma^2+2\gamma)}-\frac{2\mathcal{C}(\gamma)^2}{5(\gamma^2+2\gamma)^2}
=\frac{1}{\gamma^2+2\gamma}\Big(\frac{9}{4}r^2-\frac{2\ln^2(\gamma+1)}{5(\gamma^2+2\gamma)(\ln2)^2}\Big)\geq0\\
\Leftarrow&\frac{9}{4}r^2\geq\frac{2(\gamma+1)}{5(\gamma^2+2\gamma)\ln^22}
\Leftarrow
\gamma\geq\frac{8}{45r^2\ln^22}.
\end{split}
\vspace*{-10pt}
\end{equation}
\end{itemize}
Combining both cases together, we have that $\Delta_5\geq0$, when \vspace*{-5pt}
\begin{equation}
    \gamma\geq \max\{ \frac{1}{5r\ln2},~ \frac{8}{45r^2\ln^22}\}.
    \vspace*{-10pt}
\end{equation}
Namely, compared to \eqref{eq:con_condi}, the above condition for $\gamma$ is  tighter but more intuitive for the joint convexity of error probability $\varepsilon$ with respect to $m$ and $\gamma$. 
\vspace*{-10pt}
\section{Proof of Lemma~\ref{lemma:joint_convex_reformulated}}
The Hessian matrix of $\varepsilon_i$ with respect to $\mathbf{a}$ and $\mathbf{b}$ can be written as\vspace*{-5pt}
\begin{align}
\mathbf{\widetilde{H}}\!&=\! \left( \! 
\begin{array}{cc} 
\frac{\partial{\varepsilon_i}}{\partial{m_i}}\frac{\partial^2{m_i}}{\partial{a}^2_i}\!+\!\frac{\partial^2{\varepsilon_i}}{\partial{m}^2_i}\left(\frac{\partial{m_i}}{\partial{a_i}}\right)^2\!&\! \frac{\partial{\varepsilon_i}}{\partial{m_i}\partial{p_i}}\frac{\partial{m_i}}{\partial{a_i}}\frac{\partial{p_i}}{\partial{b_i}} \\ 
\frac{\partial{\varepsilon_i}}{\partial{m_i}\partial{p_i}}\frac{\partial{m_i}}{\partial{a_i}}\frac{\partial{p_i}}{\partial{b_i}} &\frac{\partial{\varepsilon_i}}{\partial{p_i}}\frac{\partial^2{p_i}}{\partial{b}^2_i}\!+\!\frac{\partial^2{\varepsilon_i}}{\partial{p}^2_i}\!\left(\frac{\partial{p}_i}{\partial{b}_i}\right)^2\!\\ 
\end{array}
\! \right)\nonumber 
.\label{eq:hessian_ab}
\vspace*{-10pt}
\end{align}
To show the convexity of $\varepsilon_i$, we first investigate the sign of upper-left element
of $\mathbf{\widetilde{H}}$:\vspace*{-5pt}
\begin{equation}
\label{eq:d^2e_db^2}
\begin{split}
    \frac{\partial^2 \varepsilon_i}{\partial \mathbf{a}^2}
    &=\frac{\partial{\varepsilon_i}}{\partial{m_i}}\frac{\partial^2{m_i}}{\partial a_i^2}\!+\!\frac{\partial^2{\varepsilon_i}}{\partial{m}^2_i}\left(\frac{\partial{m_i}}{\partial{a_i}}\right)^2
    =\frac{1}{a^4_i}\frac{\partial \varepsilon^2_i}{\partial m^2_i}+\frac{2}{a^3_i}\frac{\partial \varepsilon_i}{\partial m_i}\\
    &=A\left(
    \frac{1}{a_i}w_i\left(\frac{\partial w_i}{\partial m_i}\right)^2-\frac{1}{a_i}\frac{\partial^2 w_i}{\partial m^2_i}-2\frac{\partial w_i}{\partial m_i}
    \right)
    =A\left(\underbrace{\frac{\partial w_i}{\partial m_i}}_{\geq 0}\left(
    \frac{w_i}{a_i}\frac{\partial w_i}{\partial m_i}- 2\right)-\underbrace{\frac{\partial^2 w_i}{\partial m^2_i}}_{\leq 0}
    \right),
\end{split}
\vspace*{-10pt}
\end{equation}
where $A=\frac{1}{a^3_i\sqrt{2\pi}}e^{-\frac{w^2}{2}}$. Therefore, the derivate is non-negative if $\frac{w_i}{a_i}\frac{\partial w_i}{\partial m_i}- 2$ is non-negative. Note that for reliable transmission, $\varepsilon_i\ll 1$, i.e., $w_i\geq Q^{-1}(0.1)\geq 1.25$. Then, we have\vspace*{-5pt}
\begin{align}
    \frac{w_i}{a_i}\frac{\partial w_i}{\partial m_i}- 2
    =&
    w_im_i(p_i+1)\sqrt{\frac{m_i}{p_i(p_i+2)}}\log(p_i+1)-2 
    +w_i(p_i+1)\sqrt{\frac{m_i}{p_i(p_i+2)}}\frac{D_i}{m_i}\nonumber \\
    \overset{p_i\geq 1}{\geq}&1.25(1+1)\sqrt{1/3}\log(2)-2\geq 0.\vspace*{-10pt}
\end{align}
\begin{figure*}
\vspace*{-10pt}
\begin{equation}
\label{eq:det_tilde_H_expand}
    \begin{split}
     \det (\mathbf{\widetilde{H}})
    =&\frac{4b^2_i}{a^4_i}\frac{\partial^2\varepsilon_i}{\partial m^2_i}\frac{\partial^2\varepsilon_i}{\partial p^2_i}
     -\frac{2b_i}{a^2_i}\left(\frac{\partial^2\varepsilon_i}{\partial m_i \partial p_i} \right)^2
     +\frac{2}{a^4}\frac{\partial^2\varepsilon_i}{\partial m^2_i}\frac{\partial \varepsilon_i}{\partial p_i}  
     +\frac{8b^2_i}{a^3_i}\frac{\partial^2\varepsilon_i}{\partial p^2_i}
     \frac{\partial\varepsilon_i}{\partial m_i}
     +\frac{4}{a^3_i}\frac{\partial\varepsilon_i}{\partial m_i}\frac{\partial\varepsilon_i}{\partial p_i}\\
     =&\frac{b_i}{a^2_i}\left(\underbrace{2\frac{b_i}{a^2_i}\frac{\partial^2\varepsilon_i}{\partial m^2_i}\frac{\partial^2\varepsilon_i}{\partial p^2_i}
     -2\left(\frac{\partial\varepsilon_i}{\partial m_i \partial p_i}\right)^2}_{x_1}\right)+\frac{4}{a^3_i}\underbrace{\frac{\partial \varepsilon_i}{\partial m_i}\frac{\partial \varepsilon_i}{\partial p_i}}_{\geq 0}\\
     &+\frac{A^2}{a^3_i}\left(4b^2_iw_i \underbrace{\left(\frac{\partial w_i}{\partial m_i} \right)^2\frac{\partial w_i}{\partial p_i}}_{\geq 0}\left(\underbrace{\frac{2w_i}{a_i}\frac{\partial w_i}{\partial p_i}-2}_{x_2}\right) \right.
      \left.+8b^2_i\underbrace{\left(-\frac{\partial^2 w_i}{\partial m^2_i}\right)\frac{\partial w_i}{\partial p_i}}_{\geq 0}\left(\underbrace{\frac{2w_i}{a_i}\frac{\partial w_i}{\partial p_i}-1}_{\geq x_2} \right)\right.
      \\
      &\left. + \frac{w_i}{a_i}\underbrace{\left(\frac{\partial w_i}{\partial p_i} \right)^2\frac{\partial w_i}{\partial m_i}}_{\geq 0}
      \left(\underbrace{\frac{2b^2_iw_i}{a_i}\frac{\partial w_i}{\partial m_i}-1}_{x_3}\right)\right. \\
      &
      \left. +\frac{2}{a_i}\underbrace{\left(-\frac{\partial^2 w_i}{\partial p^2_i}\right)\frac{\partial w_i}{\partial m_i}}_{\geq 0}\left(\underbrace{\frac{3b^2_iw_i}{a_i}\frac{\partial w_i}{\partial m_i}-1}_{\geq x_3} \right)+\frac{12b^2_i}{a_i}\underbrace{\frac{\partial^2 w_i}{\partial p^2_i}\frac{\partial^2 w_i}{\partial m^2_i}}_{\geq 0}
      \right)\\
      \geq&\frac{b_i}{a^2_i} x_1+\frac{A^2}{a^3_i}
        \left( 
            \left(
                \left(\frac{\partial w_i}{\partial m_i} \right)^2\frac{\partial w_i}{\partial p_i}
                -\frac{\partial^2 w_i}{\partial m^2_i}\frac{\partial w_i}{\partial p_i} 
            \right)x_2 +
            \left(
                \left(\frac{\partial w_i}{\partial p_i} \right)^2\frac{\partial w_i}{\partial m_i}
                -\frac{\partial^2 w_i}{\partial p^2_i}\frac{\partial w_i}{\partial m_i}  
            \right)x_3
        \right) 
     \end{split}
\end{equation}
\hrule
\end{figure*}
Hence, the upper-left element of $\mathbf{\widetilde{H}}$, i.e., the second-order derivative of $\varepsilon_i$ with respect to $\mathbf{a}$ is non-negative. 

Next, 
we more specifically investigate the
determinant of $\mathbf{\widetilde{H}}$ in~\eqref{eq:det_tilde_H_expand}, in which three terms $x_1$, $x_2$ and $x_3$ are defined to facilitate further analysis. According to Lemma~\ref{lemma:joint_convex}, the error probability $\varepsilon_i$ is jointly convex in $m_i$ and $\gamma_i=p_i$. Based on the above convexity, respectively for each $x_j$, $j\in\{1,2,3\}$, we have\vspace*{-10pt}
\begin{equation}
    x_1
    \geq 2
        \left( 
           \frac{\partial^2\varepsilon_i}{\partial m^2_i}\frac{\partial^2\varepsilon_i}{\partial p^2_i}
            -\left(\frac{\partial^2\varepsilon_i}{\partial m_i \partial p_i}\right)^2
         \right)
     \overset{{\rm Lemma}~\ref{lemma:joint_convex}}{\geq} 0,\vspace*{-10pt}
\end{equation}\vspace*{-10pt}
\begin{align}
    x_2
    \geq& w_i\sqrt{\frac{m_i}{{1-\frac{1}{(1+p_i)^2}}}}
        \left(
            \frac{1}{p_i+1}-\frac{\log(p_i+1)}{(p^2_i+2p_i)(p_i+1)}
        \right)\!-\!1 
    \geq 1.25\frac{\sqrt{6}}{2}(\frac{1}{2}-\frac{\log 2}{6})-1\geq 0,
\vspace*{-10pt}
\end{align}
and\vspace*{-10pt}
\begin{equation}
    x_3 \geq \frac{w_i}{a_i}\frac{\partial w_i}{\partial m_i}- 2\geq 0.\vspace*{-10pt}
\end{equation}
Therefore, we have $x_j\geq 0, \forall j=1,2,3$ and all other components in~\eqref{eq:det_tilde_H_expand} are non-negative, i.e., it holds that $\det \mathbf{\widetilde{H}} \geq 0$.

As a result, $\varepsilon_i$ is jointly convex in $\mathbf{b}$ and $\mathbf{a}$. 


\bibliographystyle{IEEEtran.bst}
\bibliography{reference_convex_letter}
\end{document}